\documentclass{amsart}
\newenvironment{IEEEproof}{\begin{proof}}{\end{proof}}

\usepackage{amsmath}
\usepackage{amsfonts}
\usepackage{amsthm}
\usepackage{amssymb}
\usepackage{mathtools}
\usepackage{mathabx,amscd}
\usepackage[utf8]{inputenc}
\usepackage[T1]{fontenc}
\usepackage{url}
\usepackage{ifthen}
\usepackage{cite}
\usepackage{caption}
\usepackage[normalem]{ulem}
\usepackage{enumerate}
\usepackage{algorithm}
\usepackage[noend]{algpseudocode}
\usepackage{tikz}
\usepackage{pgfplots}
\usepackage{nicematrix}
\usetikzlibrary{matrix,arrows,decorations.pathmorphing}
\usetikzlibrary{patterns}
\usepackage{subcaption}
\usetikzlibrary{decorations.pathreplacing}
\usetikzlibrary{arrows.meta, positioning}
\usepackage{xcolor}
\usepackage[hidelinks,pdfencoding=auto]{hyperref}
\usepackage{mathrsfs}
\hypersetup{colorlinks,breaklinks,
            citecolor=cyan,
             urlcolor=blue,
             linkcolor=blue}
\usepackage{colortbl}
\usepackage{fullpage}
\usepackage{booktabs}
\usepackage{float}
\theoremstyle{plain}
\newtheorem{thm}{Theorem}[section]
\newtheorem{lem}[thm]{Lemma}
\newtheorem{pro}[thm]{Proposition}

\theoremstyle{definition}
\newtheorem{defn}[thm]{Definition}

\newtheorem{rem}[thm]{Remark}

\newtheorem{assumption}[]{Assumption}

\usepackage[most]{tcolorbox}
\xdefinecolor{myblue}{rgb}{0.1,0.1,0.8}
\xdefinecolor{mycol1}{rgb}{0.01,0.7,0.7}
\xdefinecolor{mycol}{rgb}{0.5,0.01,0.01}
\xdefinecolor{mycol2}{rgb}{0.01,0.7,0.01}

\xdefinecolor{yell}{rgb}{0.03,0.04,0.04}
\xdefinecolor{orn}{rgb}{0.04,0.03,0.04}
\xdefinecolor{rd}{rgb}{0.04,0.04,0.03}

\renewcommand{\th}{\theta}
\renewcommand{\TH}{\boldsymbol{\th}}

\newcommand{\LL}{\mathbb{L}}

\newcommand{\F}{\mathbb{F}}
\newcommand{\Fq}{\F_q}
\newcommand{\Fqq}{\F_{q^2}}
\newcommand{\Fqm}{\F_{q^m}}

\newcommand{\K}{\mathbb{K}}
\newcommand{\ZZ}{\mathbb{Z}}
\newcommand{\NN}{\mathbb{N}}

\renewcommand{\a}{\alpha}

\newcommand{\code}[1]{\mathscr{#1}}
\newcommand{\CC}{\code{C}}
\newcommand{\DD}{\code{D}}

\newcommand{\RM}{\text{RM}}
\newcommand{\plotkin}{\diamondsuit}
\newcommand{\plotkinAS}{\diamond}

\renewcommand{\aa}{\mathbf{a}}
\newcommand{\av}{\mathbf{a}}
\newcommand{\bb}{\mathbf{b}}

\newcommand{\ev}{\mathbf{e}}
\newcommand{\fv}{\mathbf{f}}
\newcommand{\mv}{\mathbf{m}}
\newcommand{\nn}{\mathbf{n}}

\newcommand{\uv}{\mathbf{u}}
\newcommand{\vv}{\mathbf{v}}
\newcommand{\xv}{\mathbf{x}}
\newcommand{\yv}{\mathbf{y}}

\newcommand{\wt}{\text{w}_{\text{H}}}
\newcommand{\Tr}{\textrm{Tr}}
\newcommand{\B}{\mathcal{B}}

\newcommand{\rank}{\text{Rk}\,}
\newcommand{\rk}{\text{Rk}\,}

\newcommand{\Am}{{\mathbf{A}}}
\newcommand{\Bm}{{\mathbf{B}}}
\newcommand{\Cm}{{\mathbf{C}}}
\newcommand{\Dm}{{\mathbf{D}}}
\newcommand{\Em}{{\mathbf{E}}}
\newcommand{\Fm}{{\mathbf{F}}}
\newcommand{\Gm}{{\mathbf{G}}}
\newcommand{\Hm}{{\mathbf{H}}}
\newcommand{\II}{\mathbf{I}}

\newcommand{\Mm}{{\mathbf{M}}}

\newcommand{\Rm}{{\mathbf{R}}}

\newcommand{\Xm}{\mathbf{X}}
\newcommand{\Ym}{\mathbf{Y}}

\newcommand{\GG}{\textrm{G}}

\newcommand{\Sm}{\mathbf{S}}
\newcommand{\Um}{\mathbf{U}}

\newcommand{\Span}[2]{{\left\langle #1 \right\rangle}_{#2}}
\newcommand{\<}{\left<}
\renewcommand{\>}{\right>}

\newcommand{\Gal}{\text{Gal}}

\renewcommand{\gg}{\texttt{g}}
\newcommand{\hg}{\texttt{h}}

\newcommand{\Prob}[1]{\mathbb{P}\left( #1 \right)}
\newcommand{\Esp}[1]{\mathbb{E}\left( #1 \right)}

\newcommand{\eqdef}{\stackrel{\text{def}}{=}}
\newcommand{\map}[4]{
  \left\{
  \begin{array}{ccc}
    #1 & \longrightarrow & #2 \\
    #3 & \longmapsto     & #4
  \end{array}
  \right.
}

\newcommand{\OO}{\mathcal{O}}
\newcommand{\compsyn}{\mathcal{C}_{\text{syn}}}
\newcommand{\compdec}{\mathcal{C}_{\text{dec}}}
\newcommand{\mul}{\mathcal{M_{\K}}}

\colorlet{known}{teal}
\colorlet{unknown}{black}

\title{Recursive decoding of binary rank Reed--Muller codes and Plotkin construction for matrix codes} 
\usepackage{amsaddr}
\author{Alain Couvreur and Rakhi Pratihar}
\thanks{Inria \& Laboratoire LIX, École Polytechnique, Institut Polytechnique
  de Paris, 91120 Palaiseau CEDEX, France}
\email{\texttt{\{alain.couvreur,rakhi.pratihar\}@inria.fr}}
\begin{document}
\maketitle

\begin{abstract}
  In 2021, Augot, Couvreur, Lavauzelle and Neri introduced a new class
  of rank metric codes which can be regarded as rank metric
  counterparts of Reed--Muller codes. Given a finite Galois extension
  $\LL / \K$, these codes are defined as some specific
  $\LL$--subspaces of the twisted group algebra $\LL [\GG]$. We
  investigate the decoding of such codes in the ``binary'' case,
  \emph{i.e.,} when $\GG = (\ZZ/2\ZZ)^m$.  Our approach takes its
  inspiration from the decoding of Hamming metric binary Reed--Muller
  codes using their recursive Plotkin ``$(u ~|~ u+v)$'' structure. If
  our recursive algorithm restricts to a specific subclass of rank metric
  Reed--Muller codes, its asymptotic complexity beats that of the
  recently proposed decoding algorithm for arbitrary rank metric
  Reed--Muller codes based on Dickson matrices. Also, this decoder is
  of completely different nature and leads a natural rank metric
  counterpart of the Plotkin construction.  To illustrate this, we
  also propose a generic Plotkin-like construction for matrix rank
  metric codes with an associate decoder, which can be applied to any
  pair of codes equipped with an efficient decoder.
\end{abstract}

\bigskip

 \setcounter{tocdepth}{1}
\tableofcontents

\section{Introduction}
\emph{Rank metric codes} were introduced by Delsarte \cite{Del} with a
combinatorial interest as sets of bilinear forms on a pair of
finite-dimensional vector spaces over a finite field $\Fq$, whereas
Gabidulin independently reintroduced a variant of rank metric codes
\cite{Gab} as $\Fqm$-linear subspaces of $\Fqm^n$ and proposed a first
decoding algorithm. Due to a broad diversity of applications of
applications of these codes, such as crisscross error correction
\cite{Rot91}, network coding \cite{KK08}, space-time coding
\cite{LK05,R15,R15a} and public-key cryptography \cite{GPT91}, the
theoretical interest to study rank metric codes is constantly rising
leading to new mathematical and algorithmic challenges.  In the finite
field setting, the very first class of efficiently decodable rank
metric codes is given by Gabidulin \cite{Gab} using $q$-polynomials of
bounded degree. This construction is actually included in a more general
framework introduced in \cite{ACLN} as subspaces of the skew group
algebra $\LL[\GG]$ for an arbitrary finite Galois extension $\LL/\K$
with Galois group $\GG = \Gal(\LL/\K)$. This point of view has been
particularly useful in defining a rank analogue of Reed--Muller codes,
also called \emph{$\TH$-Reed–-Muller codes} in \cite{ACLN}, where
$\TH = (\theta_1, \ldots , \theta_m)$ specifies a generating set of
the abelian group $\GG$ (See \S \ref{RMcodes} for a discussion on
Reed--Muller-type codes in different metrics).

Towards various applications, for instance post--quantum cryptography,
it is particularly important for a family of codes to have efficient
decoding algorithms. Indeed, on one hand, McEliece-like schemes
\cite{M78} require codes with an efficient decoder and whose structure
can be hidden to the attackers.  Such a scheme has been instantiated
with codes both in Hamming (see for instance
\cite{M78,N86,ABCCGLMMMNPPPSSSTW20}) and in rank metric (see for
instance \cite{GPT91,GMRZ13,ROLLO,L17}).  On the other hand, many
Alekhnovich-like schemes in code or lattice based cryptography require
a decoder to conclude the decryption phase and get rid from a residual
noise term. See for instance \cite{AABBBBDDGLPRVZ22,A19b}. For these
reasons, there is a strong motivation in broadening the diversity of
decodable codes in rank metric: first to improve our understanding of
decoding problems and second in view toward applications to future new
cryptographic designs. However, the known classes of rank metric codes
with efficient decoding algorithms are only the following few; simple
codes \cite{SKK11, GHPT17}, some families of MRD codes including
Gabidulin codes \cite{Loid} and its variants, cf. \cite[Chapter
2]{BHLPRW}, low-rank parity-check (LRPC) codes \cite{AGHRZ} and the
interleaved version of the aforementioned codes \cite{SB10,RJB}.  In
\cite{ACLN}, a construction of rank metric analogues of Reed--Muller
codes is proposed and, for decoding, the authors suggested the use of the
\emph{rank error--correcting pairs} paradigm \cite{MP}. This led to a
polynomial--time decoder that could correct an amount of errors that remained
far below half the minimum distance.
In \cite{CP25}, the authors of the present article give a
deterministic decoding algorithm for $\TH$-Reed--Muller codes that
involved so--called \emph{$\GG$-Dickson matrices} and which corrects
any error patter of rank up to half the minimum rank distance. In the present
article, we investigate a recursive structure of these codes in the
``binary--like'' case, \emph{i.e.}, when $\GG\cong (\ZZ/{2 \ZZ})^{m}$ and
propose a new decoding technique resting on this recursive structure.

It is worth noting that the structure we identified can be considered
as a rank metric analogue to the recursive structure that the Hamming
metric binary Reed--Muller codes possess. More precisely, in the
Hamming setting, any codeword of $\RM_2(r,m)$ can be written as
$(u ~|~ u+v)$ where $u \in \RM_2(r, m-1)$ and $v\in
\RM_2(r-1,m-1)$. The general $(u~|~u+v)$ construction was introduced
by Plotkin \cite{MS77} and known as \emph{Plotkin construction} for
linear codes with Hamming metric. This recursive construction and its
use for decoding have been studied extensively. Besides binary
Reed--Muller codes and their decoding, the Plotkin construction has
striking applications since it naturally appears in the construction
of polar codes or in the post--quantum signature \emph{Wave}
\cite{BCCCDGKLNSST23,DST19a}.  One can also mention that this
construction iteratively applied on Reed--Solomon codes permits to
achieve the capacity of the discrete symmetric channel as proved in
\cite{MCT17}.

\subsection*{Our contributions}
In the present article, we demonstrate a recursive structure of
binary--like rank metric Reed--Muller codes and give a decoding
algorithm that can correct up to half minimum distance. This algorithm
may fail on some unlikely instances but has a better asymptotic
complexity than the Dickson-based decoding of \cite{CP25}. In short,
for $N \eqdef |\GG|$, the aforementioned Dickson-based approach
corrects $t$ errors $\OO(tN^2)$ operations in $\LL$ while the decoder
proposed in this article, requires $\OO(t^{\omega - 1}N)$ operations in
$\LL$, where $\omega$ is the complexity exponent of linear algebra.

Next, inspired by the structure of binary--like rank metric
Reed--Muller codes, we adapt the recursive structure to propose
Plotkin-like construction for matrix rank metric codes over any field.
As a consequence, we provide new efficiently decodable matrix rank
metric codes over finite fields which are not equivalent to Gabidulin
codes. This new framework provides new families of rank metric codes
equipped with a decoder. To our knowledge, this is the first
successful attempt of providing a rank metric counterpart of the
Plotkin construction.

  \subsection*{Organization of the article} Section \ref{sec:3}
  introduces the notation used in this paper, as well as basic
  notions regarding rank metric codes as subspaces of skew group
  algebras including binary rank metric Reed--Muller codes. In Section
  \ref{Binary_rankRM}, we discuss a recursive structure of binary rank
  metric Reed--Muller codes. A decoding algorithm for these codes
  resting on the recursive structured is presented in Section
  \ref{decoding}. We show that the recursive decoding approach permits
  to correct error patterns of rank up to half the minimum distance,
  whenever some folding of the error matrix preserves the rank. The
  complexity of this novel decoder is studied in depth in
  Section~\ref{sec:complexity}. Finally, Section \ref{sec:5} presents
  a rank metric analogue of Plotkin construction by adapting the
  recursive structure of binary rank metric Reed--Muller codes for
  matrix rank metric codes over finite fields. We demonstrate how the
  recursive decoding algorithm can be adapted to efficiently decode a
  Plotkin construction and provides new efficiently decodable
  matrix rank metric codes inequivalent to Gabidulin codes.

\subsection*{Note} The present article is an extended version of a short paper presented
at the \emph{IEEE International Symposium on Information Theory} (ISIT) 2025 \cite{CP25b}.
Compared to the original short paper, the current article includes detailed proofs and 
improves the main algorithm by simplifying its initialisation step.

 \section*{Acknowledgements} The authors warmly thank Martin Weimann
and Cl\'ement Pernet for their very useful advice. The authors are
partially funded by French Agence Nationale de la Recherche through
the France 2030 ANR project ANR-22-PETQ-0008 PQ-TLS, by French Grant
ANR projet \emph{projet de recherche collaboratif}
ANR-21-CE39-0009-BARRACUDA, and by Horizon-Europe MSCA-DN project
ENCODE.

 \section{Preliminaries}\label{sec:3}

\subsection{Notation}
Throughout this paper, $\K$ denotes a field, 
and $\LL$ denotes a finite Galois extension of $\K$.  We use $\GG$ to
denote the Galois group $\Gal(\LL/\K)$ and the elements of $\GG$ are
usually denoted as $\gg_0, \gg_1, \ldots{}, \gg_{n-1}$. By $\B$, we
denote a basis of the finite dimensional vector space $\LL$ over
$\K$. For a $\K$-linear space $V$, the span of vectors
$\boldsymbol{v}_1, \dots, \boldsymbol{v}_t \in V$ is denoted as $\Span{\boldsymbol{v}_1, \dots,
  \boldsymbol{v}_t}{\K}$. The space of linear endomorphisms of $V$ is denoted
${\rm End}_{\K}(V)$. The space of matrices with $m$ rows and $n$ columns with entries in
$\K$ is denoted using $\K^{m \times n}$ and $\LL^n$ denotes the space
of vectors of length $n$ over $\LL$. Matrices are usually denoted in
bold capital letters and we use $\rk(\mathbf{A})$ to denote rank of a
matrix $\mathbf{A}$. 
Finally, when handling complexities we will use Landau notation for comparison.
Namely, for $m$ going to infinity we denote
\begin{align*}
  f(m) = \OO(g(m)) \quad &\text{if}\quad \exists M >0, \quad\text{such that} \quad \forall m \geq M,\ f(m)\leq \kappa g(m)\quad \text{for some }\kappa >0;\\
  f(m) = \Omega(g(m)) \quad &\text{if}\quad \exists M >0, \quad\text{such that} \quad \forall m \geq M,\ f(m)\geq \kappa g(m)\quad \text{for some }\kappa >0;\\
  f(m) = \Theta(g(m)) \quad &\text{if both}\quad f(m) = \OO(g(m)) \quad \text{and}\quad f(m) = \Omega(g(m));\\
  f(m) = \circ (g(m)) \quad &\text{if}\quad f(m) = g(m)\varepsilon_m \quad \text{where}\quad \varepsilon_m \rightarrow 0.
\end{align*}
Also we denote $f(m) = \widetilde{\OO}(g(m))$ if $f(m)=\OO(g(m)P(\log(m)))$ for some polynomial $P$.

\medskip

In this section, we recall the relevant definitions and basic notions
of rank metric codes as well as their various equivalent
representations. We also record some results about Reed--Muller codes
with rank metric from \cite{ACLN}.
\subsection{Matrix codes}
Delsarte introduced rank metric codes in \cite{Del} as $\K$-linear
subspaces of the matrix space $\K^{m \times n}$ where the rank
 distance of two codewords (\emph{i.e.},
matrices) $\mathbf{A}, \mathbf{B} \in \K^{m \times n}$ is given by
\[
d_{\rk}(\mathbf{A},\mathbf{B}) = \rank(\mathbf{A}-\mathbf{B}).
\]
Such matrix spaces are called \emph{matrix rank metric codes} and
denoted by \emph{$[m \times n, k,d]_{\K}$--codes} where $k$ denotes the
$\K$-dimension of the code and $d$ denotes the minimum distance, \emph{i.e.} the minimum of the rank distances of any two distinct codewords.  The \emph{dual} of a matrix rank metric code $\CC \subseteq \K^{m \times n}$ is the matrix rank metric code 
\[
\CC^{\perp} \eqdef \{\Am \in \K^{m \times n} \colon \Tr(\Am\Cm^\top) = 0 \text{ for any } \Cm \in \CC\}. 
\]
\begin{rem}
  Note that a more abstract point of view can be adopted by
  considering subspaces of the space of $\K$-linear maps from a finite
  dimensional $\K$-linear space $V$ to another $\K$-linear space
  $W$. This point of view is somehow considered in the sequel when we
  deal with subspaces of skew group algebras (see
  \S~\ref{sec:LG-codes}).
\end{rem}

\subsection{Vector codes}
Since the works of Gabidulin \cite{Gab}, the classical literature
on rank metric codes also involves $\LL$-linear subspaces
of $\LL^n$, where the \emph{$\K$--rank} or \emph{rank} (we will omit the field $\K$ when it is
clear from the context) of a vector is defined as
\[
\forall \aa \in \LL^n,\quad  \rk_{\K}(\aa) \eqdef \dim_{\K} \Span{a_1, \dots, a_n}{\K}.
\]
Next, the distance between two vectors $\aa, \, \bb \in \LL^n$ is
defined as
\[
d_{\rk}(\aa,\bb) \eqdef \rk_{\K}(\aa-\bb).
\]
$\LL$--subspaces of $\LL^n$ are called \emph{vector rank metric codes}
and denoted by $[n, k, d]_{\LL/\K}$ codes where $k$ denotes the
$\LL$--dimension of the code and $d$ denotes the minimum distance.

It is well--known that such vector codes actually can be turned into
matrix codes by choosing a $\K$--basis $\B = ({\beta}_1, \dots, {\beta}_m)$ of
$\LL$ and proceeding as follows. Given an element $x$ of $\LL$ denote by $x^{(1)}, \dots, x^{(m)}$ its coefficients in the basis $\B$, \emph{i.e.},
$x = x^{(1)}{\beta}_1 + \cdots + x^{(m)} {\beta}_m$. Consider the map
\[
  \text{Exp}_\B : \map{\LL^n}{\K^{m \times n}}{(x_1, \dots, x_n)}{
    \begin{pmatrix}
      x_1^{(1)} & \cdots & x_n^{(1)}\\
      \vdots &  & \vdots \\
      x_1^{(m)} & \cdots & x_n^{(m)}      
    \end{pmatrix}.
}
\]
Then, any vector code $\CC \subset \LL^n$ can be turned into a matrix
code by considering $\text{Exp}_{\B}(\CC)$. The induced matrix code
depends on the choice of the basis $\B$ but choosing another basis provides
an isometric code with respect to the rank metric.

\begin{rem}
  Note that if an $\LL$--linear rank metric code can be turned into a
  matrix code, the converse is not true. A subspace of
  $\K^{m\times n}$ can be turned into a $\K$-linear subspace of $\LL^n$
  by applying the inverse map of $\text{Exp}_\B$ but the resulting
  code will not be $\LL$--linear in general. A matrix rank metric code $\CC \subseteq \K^{m \times n}$ is $\LL$-linear if its \emph{left idealiser} contains a subring isomorphic to $\LL$ \cite{LTZ17}, where the left idealiser \[L(\CC) \eqdef \{\Am \in \K^{m \times m} \colon \Am \Cm \in \CC \text{ for all } \Cm \in \CC \} \]
  is a ring with respect to sum and product of matrices.
 Thus, codes of the form
  $\text{Exp}_\B(\CC)$ when $\CC$ ranges over all $\LL$--subspaces of
  $\LL^n$ form a proper subclass of matrix codes in $\K^{m\times n}$.
\end{rem}

\subsection{Rank metric codes as
  \texorpdfstring{$\LL[\GG]$}{}--codes}\label{sec:LG-codes}

The study of rank metric codes as $\LL$-subspaces of the skew group algebra
$\LL[\GG]$ has been initiated in \cite{ACLN}.
It generalizes the study of rank metric codes over arbitrary cyclic Galois extensions in
\cite{ALR18}.
We recall below the definitions and
basic notions of rank metric codes in this setting.

Consider an arbitrary but fixed finite Galois extension $\LL/\K$ with
$\GG \eqdef \Gal(\LL/\K)$. The \emph{skew group} algebra $\LL[\GG]$ of $\GG$ over $\LL$ is defined as
\[
\LL[\GG] \eqdef \Bigg\{\sum_{\gg \in \GG} a_\gg \gg ~\colon~ a_\gg \in \LL\Bigg\}
\]
and endowed with its additive group structure and the following composition law derived from the group law of $\GG$:
\[
(a_\gg \gg ) \circ (a_\hg \hg) = (a_\gg \gg(a_\hg)) (\gg \hg),
\]
which is extended by associativity and distributivity.
This equips $\LL[\GG]$ with a non-commutative algebra structure.

\begin{thm}\label{thm:correspondence_poly_endo}
  Any element $A = \sum_\gg a_\gg \gg \in \LL[\GG]$ defines a
  $\K$-endomorphism of $\LL$ that sends $x \in \LL$ to
  $\sum_\gg{a_\gg\gg(x)}$. This correspondence induces
  a $\K$-linear isomorphism between $\LL[\GG]$ and
  $\ensuremath{\textrm{End}}_{\K}(\LL)$.
\end{thm}

\begin{proof}
  It follows from Artin's lemma on linear independence of characters. A proof can be found, for instance, in {\cite[Thm.~1]{GQ09}}. \end{proof}
Thus, the \emph{rank} of an element $A \in \LL[\GG]$ is well-defined as its rank when viewed as a $\K$-linear endomorphism of $\LL$.
From the above theorem, it is clear that with respect to a fixed basis
$\B = ({\beta}_1, \dots, {\beta}_m)$ of $\LL/\K$, we get $\K$--linear rank--preserving isomorphisms
\begin{equation}\label{representation}
  \LL[\GG] \cong \textrm{End}_\K(\LL) \cong \K^{m\times m}.
\end{equation}
Also, w.r.t. the basis $\B = \{ {\beta}_1, \ldots, {\beta}_m\}$, every element  $A \in \LL[\GG]$ can be seen as a vector
\[
  \textbf{a} = (A({\beta}_1), \ldots, A({\beta}_m)) \in \LL^m.
\]
The aforementioned definition of rank for a vector of $\LL^m$
coincides with the rank of $A$ when regarded as a $\K$--endomorphism
of $\LL$ (according to Theorem~\ref{thm:correspondence_poly_endo}).

\begin{rem}
In the particular case of a Galois extension of finite fields
$\LL/\K$, the group $\GG$ is cyclic and there
are many characterizations of $\LL[\GG]$ studied in \cite{WL13}. One of
the very well--known characterization is in terms of \emph{linear
polynomials} studied by Ore \cite{Oreqpoly} followed by his work on the
theory of non-commutative polynomials \cite{Orenoncom}.  Let
$\K = \Fq$ and $\LL = \Fqm$ for some prime power $q$ and a positive
integer $m$, then the \emph{linear polynomials} over $\Fqm$ are given by
\[
  L(x) = \sum\limits_{i=0}^{d} a_ix^{q^i},\ \  \text{for\ some\ } d \in \NN \quad
  \text{and}\quad a_0, \dots, a_d \in \Fqm
\]
and endowed with the composition law to give a structure of (non
commutative) ring. With this point of view, the skew group algebra
$\Fqm[\GG]$ is isomorphic to the ring of linear polynomials modulo
the two--sided ideal generated by $x^{q^m}-x$.
\end{rem}

\begin{defn}
    An \emph{$\LL$-linear rank metric code} $\CC$ in the skew group algebra $\LL[\GG]$ is an $\LL$-linear subspace of $\LL[\GG]$, equipped with the rank distance. The dimension of $\CC$ is defined as its dimension as $\LL$-vector space. The minimum rank distance is defined as 
    \[
    d(\CC) \eqdef \min \{\rk(A) ~\colon~ A \in \CC\setminus \{0\}\},
  \]
  where the rank of $A \in \LL [\GG]$ is the rank of the
  $\K$--endomorphism it induces on $\LL$ (according to
  Theorem~\ref{thm:correspondence_poly_endo}).
  \end{defn}
  
  We denote the parameters of an $\LL$-linear rank metric code
  $\CC \subseteq \LL[\GG]$ of dimension $k$ and minimum distance $d$ by
  $[\LL [\GG],k,d]$. If $d$ is either unknown or clear from the context, we
  simply write $[\LL[\GG],k]_{\K}$.

As observed earlier, an element of $\LL[\GG]$ can be seen as a
  $\K$-linear endomorphism of $\LL$. Therefore, if we fix a $\K$-basis
  $\B$ of $\LL$, then,  one can
  transform an $[\LL[\GG],k,d]_{\K}$--code $\CC$ into an
  $[m \times m, k, d]_{\K}$ matrix code or into an $[m,k,d]_{\LL/\K}$ vector code.

\subsection{Rank metric Reed--Muller codes}\label{RMcodes}

The class of classical $q$-ary Reed--Muller codes is among the most widely studied class of linear codes in the Hamming metric. Introduced independently by Muller and Reed in 1954 in the binary case (i.e., $q=2$), the class of $q$-ary Reed--Muller codes $\textrm{RM}_{q}(r,m)$ of order $r$ and type $m$ was generalized by Kasami, Lin and Peterson \cite{KLP68,KLP68b} in the 1960's as the code over $\Fq$ of length $q^m$ obtained by evaluation of polynomials in $m$ variables of degree bounded by $r$ at the points of $\F_q^m$. An extension to the projective setting is proposed and studied in \cite{L88,L90,S91}.

Counterparts of Hamming metric Reed--Muller codes in the rank and the sum-rank metric have been defined by evaluations of certain multivariate versions of non-commutative Ore polynomials \cite{Orenoncom}. For an introduction to the multivariate extension of the algebra of Ore polynomials that has been used to define the sum-rank metric counterpart of Reed--Muller codes, refer to \cite[\S~1]{BC24}. An iterated construction of multivariate skew polynomials leading to Reed--Muller codes in skew metric is studied in \cite{GU19}.

For the present work, we recall the construction of \cite{ACLN} using skew group algebras.
The so--called \emph{$\TH$-Reed--Muller codes} of \cite{ACLN} can be
considered as a ``multivariate'' version of the Gabidulin codes
defined from cyclic extensions in \cite{ALR18}. Let
$\LL/\K$ be an abelian extension with
$\GG = \ZZ/{n_1 \ZZ} \times \dots \times \ZZ/{n_m \ZZ}$ with a system
of generators $\theta_1, \ldots, \theta_m$, a \emph{$\TH$-monomial}
$\theta_1^{i_1}\cdots\theta_m^{i_m}$ describes the $m$-tuple
$(i_1, \ldots, i_m) \in \ZZ/{n_1 \ZZ} \times \dots \times \ZZ/{n_m
  \ZZ}$. Every element in $\LL[\GG]$ has a unique
representation as a \emph{$\TH$-polynomial}
\[P = \sum_{(i_1,\ldots, i_m)} b_{(i_1,\ldots,
    i_m)}\theta_1^{i_1}\cdots\theta_m^{i_m}.
\]
We define
\[
  \deg_{\TH}(P) \eqdef\max\{i_1+ \cdots+i_m ~\colon~ b_{(i_1, \ldots,
    i_m)}\neq 0\} .
\]
For $0< r \le \sum_{i=1}^m(n_i -1)$,
the \emph{$\TH$-Reed--Muller code of order $r$ and type
$\boldsymbol{n} \eqdef (n_1,\ldots, n_m)$} is defined as
\[
  \RM_{\TH}(r,\nn)\eqdef \{P \in \LL[\GG] ~\colon~ \deg_{\TH}(P) \le
  r\}.
\] By fixing a $\K$-basis $\beta = \{\beta_1, \ldots, \beta_N \}$ of
$\LL$, the code $\RM_{\TH}(r,\nn)$ can be seen in the vector form as
\[\{(P(\beta_1), \ldots, P(\beta_N)) ~\colon~ P \in \LL[\GG], \,
  \deg_{\TH}(P)\le r\} \subseteq \LL^N,\] where $N = |\GG|$ and
equivalently, every codeword can be seen as $N\times N$ matrices with
entries in $\K$. When $m =1$, \emph{i.e.} $\GG$ cyclic, we recover
Gabidulin codes \cite{Del,Gab} when $\LL$ is finite and their
generalisation to a cyclic extension \cite{ALR18} when $\LL$ is
arbitrary.  

\begin{rem}
It is worth noting that contrary to Gabidulin codes,
$\TH$--Reed--Muller codes cannot be defined over finite fields as soon as the underlying Galois group is not cyclic. Rather, $\TH$--Reed--Muller codes can be defined, e.g., over abelian extensions of number fields, local fields, function fields of algebraic curves over finite fields.  
\end{rem}

\subsection{Binary-like rank metric Reed--Muller codes.}

In the present article, we fix a positive integer $m$ and focus on the
particular case of rank metric Reed--Muller codes with
$\GG \cong (\ZZ/2\ZZ)^m $ and we call these \emph{binary rank metric
  Reed--Muller codes} since the $\TH$-polynomials resemble those in the
case of binary Reed--Muller codes with Hamming metric. In particular,
binary rank metric Reed--Muller codes have exactly the same parameters
as their Hamming metric counterpart. See \cite[Prop.~48 \& Thm.~50]{ACLN}.
We denote by $N$ the
``code length'':
\[
  N \eqdef 2^m.
\]

\begin{defn}\label{def:RM_codes}
  Let $r,m$ be positive integers such that $m > 1$. A
  \emph{rank metric binary Reed--Muller code} of length $N=2^m$ over
  the extension $\LL/\K$ of order $r$ and type $m$ is the following
  $\LL$-subspace of the skew group algebra $\LL[\textrm{G}]$:
  \[
  \RM_{\LL/\K}(r,m) \eqdef \bigg\{f \in \LL[\textrm{G}] ~\colon~ f =\sum_{\gg\in \textrm{G},\ \wt(\gg) \leq r} f_{\gg}
    \gg \bigg\},
    \]
    where $\wt(\gg)$ is the Hamming weight of the element $\gg$ seen
    as a vector in $(\ZZ/2\ZZ)^m$.  By convention, for any $r<0$, we
    define $\RM_{\LL/\K}(r,m)$ to be the zero subspace of $\LL [\GG]$.    
\end{defn}

\begin{rem}
    Note that whenever $r \geq m$, then $\RM_{\LL/\K}(r,m) = \LL[\GG]$.
\end{rem}

\begin{rem}
To keep in mind the analogy with binary Hamming metric Reed--Muller
codes, we denote $\RM_{\LL/\K}(r,m)$ with $m$ instead of
$\textbf{n} = (2, \ldots, 2)$ as generally defined (see Section~\ref{subsec:classical_RM}
where Hamming metric binary Reed--Muller codes are shortly discussed). Moreover, we note the
extension $\LL/\K$ as a subscript and omit the mention of a  set of
generators $\TH$. This choice is motivated by the fact that we will
handle various field extensions and hence it will be useful to keep
track of which extension we are considering.
\end{rem}

 \section{Recursive structure of binary rank metric Reed--Muller
  codes}\label{Binary_rankRM}
 In this section, we show a recursive structure of
  binary rank metric Reed--Muller codes analogous to the classical case
  of Hamming metric Reed--Muller codes. First, to motivate this work, we recall the
  recursive structure of classical Reed--Muller codes in the Hamming metric setting and
  how this structure is used for decoding.

  \subsection{The Hamming case}\label{subsec:classical_RM}
   Recall that the binary Reed--Muller codes of order
    $r$ and type $m$ are defined as the evaluation of polynomials in
    $\F_2[x_1, \ldots, x_m]$ at all the vectors in $\F_2^{m}$. For a fixed
    ordering of the $2^m$ elements of $\F_2^m$ in the reverse lexicographic
    order, the code $\RM_2{(r,m)}$ is of the form
$$\{(f(p_1), \ldots, f(p_{2^m})): f \in \F_2[x_1, \ldots, x_m], \deg \,
f \leq r \}.$$ A crucial structure of these codes comes from the
observation that any polynomial $f \in \F_2[x_1, x_2, \ldots, x_{m}]$
of degree $r$ can always be decomposed into two parts:
\[f(x_1,\ldots,x_m) = g(x_1,\ldots,x_{m-1})+x_m
  h(x_1,\ldots,x_{m-1}),\] where $\deg g \leq r$ and
$\deg h \leq r-1$.
A consequence of these observations is
the following recursive structure of binary Reed--Muller codes, first
studied by Plotkin.

\begin{lem}[Plotkin, 1960]\label{recursivestructure}
  When sorting the elements of $\F_{2}^m$ in the lexicographic order,
  the binary Reed--Muller code $\RM_2(r,m)$ can be constructed
  recursively using smaller order Reed--Muller codes: 
\begin{equation*}
{\RM_2(r,m) = \{(\uv ~|~ \uv+\vv) ~\colon~ \uv \in \RM_2(r,m-1),\, \vv
  \in \RM_2(r-1,m-1) \}.}
\end{equation*}

\end{lem}
 The recursive structure permits to prove that the minimum
distance of $\RM_2(r,m)$ is $2^{m-r}$ and leads to an
efficient decoding algorithm correcting up to half their minimum
distance in quasi-linear time in the block length \cite{ASAM}.
This decoder is shortly discussed in the beginning of Section \ref{decoding}.

\subsection{The recursive structure of rank metric Reed--Muller codes}

 To explore the recursive structure of the $\RM_{\LL/\K}(r,m)$, first we
note that if $\GG \cong (\ZZ/2\ZZ)^m$, the Galois extension $\LL/\K$
can be of the following two types.

\begin{lem}\label{extensiontype}
  Let $\LL/\K$ be a Galois extension with
  $\GG = \Gal(\LL/\K) \cong (\ZZ/{2\ZZ})^m$ where $m \ge 1$. Then
  there exist $\a_1, \dots, \a_m \in \LL$ such that
  $\LL = \K(\a_1, \ldots, \a_m)$ where $\K(\a_1),\dots, \K(\a_m)$ are
  quadratic extensions over $\K$. Moreover, for any $i$,
  the minimal polynomial of $\a_i$ has the following shape:
\[
\begin{cases}
    X^2 - a_i \text{ for some } a_i \in \K \text{ when } {\rm Char } (\K) \neq 2,\\
    X^2 +X + a_i \text{ for some } a_i \in \K \text{ when } {\rm Char } (\K) = 2.
\end{cases}
\]
\end{lem}

\begin{IEEEproof}
  Let $\GG = \<\th_1, \ldots, \th_m \>$ and $\LL^{H_i}$ is the fixed
  field of
  $H_i = \langle \th_1, \ldots, \widehat{\th_i}, \ldots, \th_m
  \rangle$ for $i = 1, \ldots, m$. Since $\bigcap_{i=1}^{m} H_i = Id$,
  the compositum of the $\LL^{H_i}$'s in $\LL$ is $\LL$
  itself. Considering the cardinalities of the $H_i$'s, we have
  $\LL^{H_i} = \K(\a_i)$  where the minimal polynomial
  of $\a_i$ is a quadratic polynomial. If
  Char $\K \neq 2$, then by completing the square, the minimal polynomial of
  $\a_i$ over $\K$
  becomes $X^2 - a_i$ for some $a_i \in \K$.
  Similarly, when $\text{Char}\ \K = 2$, the minimal polynomial of $\alpha_i$
  over $\K$ is a quadratic polynomial $X^2 + sX + t$ for some $s, t \in \K$.
  Then, replacing $\alpha_i$ by $s^{-1}\alpha_i$ we get a minimal polynomial
  $X^2 + X + s^{-2}t$.
\end{IEEEproof}
The case of characteristic $\neq 2$ will be referred to as the \emph{Kummer} case
since all the minimal intermediary extensions are Kummer ones and the
case of characteristic $2$ will be referred to as the \emph{Artin--Schreier}
case.

\subsubsection{The Kummer case}
First, we consider the case $\text{Char} \,\K \neq 2$, \emph{i.e.}, a Kummer
type extension $\LL/\K$.  We set $\LL = \K (\a_1, \dots, \a_m)$ where
each one of the $\a_i$'s satisfies $\a_i^2 = a_i$ for some $a_i \in \K$.
We also introduce the intermediary extension
\[
  \LL_{m-1} \eqdef \K (\a_1, \dots, \a_{m-1}).
\]
The Galois group $\GG$ has generators $\theta_1, \dots, \theta_m$ which satisfy
\[
  \forall i,j \in \{1,\dots, m\},\quad
  \theta_i(\alpha_j) = (-1)^{\delta_{ij}} \alpha_j,
\]
where $\delta_{ij}$ denotes the Kronecker Delta. Finally, we set
\[
  \GG_{m-1} \eqdef \langle \theta_1, \dots, \theta_{m-1}\rangle
\]
and note that $\Gal (\LL_{m-1}/\K) \cong \GG_{m-1}$.
This is summarized in the diagram below.

  \begin{center}
    \begin{tikzpicture}
      \node at (2.35,3) {$\LL = \K (\a_1, \dots, \a_m)$}; \node at
      (2.85,1.5) {\(\LL_{m-1} = \K (\a_1, \dots, \a_{m-1})\)}; \node at
      (1,0) {\(\K\)}; \draw[-] (1,.3) to (1,1.2) ; \draw[-] (1,1.8) to
      (1,2.7) ; \draw[-] (1.1,1.8) .. controls (1.2,2.1) and (1.2,2.4)
      .. (1.1,2.7); \node at (1.6,2.25)
      {$\scriptstyle{\<\theta_m\>}$}; \draw[-] (1.1,.3) .. controls
      (1.2,.6) and (1.2,.9) .. (1.1,1.2); \node at (2.2,.7)
      {$\scriptstyle{\GG/ \<\theta_m\> \cong\GG_{m-1}}$}; \draw[-] (.9,.3)
      .. controls (.6,.6) and (.6,2.4) .. (.9,2.7); \node at (0.3,1.5)
      {$\scriptstyle{\GG}$};
    \end{tikzpicture}
  \end{center}

  In addition, we will equip $\LL$ and its intermediary subfields with
  a specific $\K$--basis that is iteratively constructed as follows. We take
  $\mathcal{B}_1 = (1, \alpha_1)$ to be a $\K$--basis of $\K(\alpha_1)$.
  Then, for any $i\in \{1,\dots, m-1\}$ the basis $\mathcal{B}_{i+1}$ is defined
  as $\mathcal{B}_i \cup \mathcal{B}_i\alpha_{i+1}$. For the sake of convenience, we set $\mathcal{B} \eqdef \mathcal{B}_m$.

The statement to follow describes the recursive structure of binary rank metric 
Reed--Muller codes. In the Hamming setting, the vectors in the codes had
a recursive structure. In the rank setting the recursive structure
can be read via their matrix representation.

\medskip

 \begin{pro}[Recursive structure in the Kummer case]\label{rec_structure}
   Let ${\rm Char ~} \K \neq 2$, $\Gal(\LL/\K) \cong (\ZZ/{2\ZZ})^m$ and $r$ be an 
   integer and $m$ be a nonnegative integer. The code $\RM_{\LL/\K}(r,m)$ of Definition~\ref{def:RM_codes} has the
   following recursive structure.
   \begin{itemize}
    \item If $r\leq -1$ then it is the zero code, composed only with the
       $2^m \times 2^m$ zero matrix. 
    \item Otherwise, it is defined as:
      \begin{equation*}
        \left\{\begin{pmatrix}
         {\Am_0+\Bm_0 } &\quad &{a_m( \Am_1-  \Bm_1)}\\ 
                && \\
          { \Am_1+ \Bm_1} &\quad &{\Am_0 -\Bm_0}
          \end{pmatrix} \quad\colon\quad  \begin{aligned}
             &\Am_i \in \RM_{{\LL_{m-1}}/\K}(r, m-1), \notag \\
             &\Bm_i \in \RM_{{\LL_{m-1}}/\K}(r-1, m-1)
            \end{aligned}\right\},
          \end{equation*}
           where $a_m \eqdef \alpha_m^2 \in \K$.
          \end{itemize}
\end{pro}

\begin{IEEEproof}
  Let us give a decomposition of elements $F = \sum_{\gg \in \GG} f_\gg \gg$ of $\LL[\GG]$. 
  By considering separately $\theta_m$ and noting that $\GG$ splits
  into the disjoint union of the two cosets $\GG_{m-1}$ and
  $\GG_{m-1} \theta_m$, we can split $F$ as
\[
  \begin{aligned}
    F  
      & = \sum_{\gg \in \GG_{m-1}} f_\gg \gg + \left(\sum_{\gg \in \GG_{m-1}} f_{\gg \theta_m} \gg\right)\theta_m.
    \end{aligned}
  \]

  Next, $\LL$ splits into $ \LL = \LL_{m-1} \oplus \LL_{m-1} \alpha_m.$
Thus, for any $a \in \LL$, let  $a = a^0 + a^1 \alpha_m$ with $a^{0}, a^{1} \in \LL_{m-1}$.
  Then, we get a splitting of $F$ as
  \begin{equation}\label{eq:recursive_1}
    F = \sum_{\gg \in \GG_{m-1}} f_\gg^0 \gg +
    \alpha_m  \sum_{\gg \in \GG_{m-1}} f_\gg^1 \gg + \left(\sum_{\gg \in \GG_{m-1}} f^0_{\gg \theta_m} \gg \right)\theta_m 
     + \alpha_m\left(\sum_{\gg \in \GG_{m-1}} f^1_{\gg \theta_m} \gg \right)\theta_m.
  \end{equation}
  Set
  \[
    A_0 = \sum_{\gg \in \GG_{m-1}} f_\gg^0 \gg,  \qquad A_1 =  \sum_{\gg \in \GG_{m-1}} f_\gg^1\gg,\qquad
    B_0 = \sum_{\gg \in \GG_{m-1}} f^0_{\gg \theta_m} \gg, \quad \text{and} \quad   B_1 = \sum_{\gg \in \GG_{m-1}} f^1_{\gg \theta_m} \gg,
  \]
then the degree constraint on $F \in \RM_{\LL/\K}(r,m)$
entails
\[
    A_0, A_1 \in \RM_{\LL_{m-1}/\K}(r, m-1) \quad \text{and} \quad
    B_0, B_1 \in \RM_{\LL_{m-1}/\K}(r-1, m-1).
  \]

  Equation~(\ref{eq:recursive_1}) can then be rewritten as 

  \begin{equation}\label{eq:recursive_2}
  \setlength{\belowdisplayshortskip}{3pt}
    F = A_0 + \alpha_m A_1 +(B_0 + \alpha_m B_1)\theta_m.
  \end{equation}  
Recall that we chose a $\K$--basis $\mathcal B$ of $\LL$ such that
  $\mathcal B = \mathcal{B}_{m-1} \cup \mathcal{B}_{m-1}\alpha_m$ where
  $\mathcal{B}_{m-1}$ is a $\K$--basis of $\LL_{m-1}$.
  By
  denoting $\Am_0, \Am_1, \Bm_0, \Bm_1$ the respective representations
  of $A_0, A_1, B_0, B_1$ in $\mathcal{B}_{m-1}$, we get the desired
  matrix form of $F$ in $\mathcal B$.
  
Indeed, as ${\theta_m}_{|\LL_{m-1}} = \text{id}_{\LL_{m-1}}$ and thus 
\[
  F_{|\LL_{m-1}} = (A_0+B_0) + \alpha_m (A_1 + B_1), 
\]
the left-hand half of the matrix in the statement corresponds to
$F_{|\LL_{m-1}}$.
Similarly, the right-hand half corresponds to the restriction of $F$ to $\LL_{m-1} \alpha_m$. To see this, let  $u \alpha_m\in \LL_{m-1} \alpha_m$
with $ u \in \LL_{m-1}$, which yields
\[
  F(u\alpha_m) =  a_m (A_1(u) - B_1(u)) + \alpha_m (A_0(u) - B_0(u))\]
and thus the shape. 
\end{IEEEproof}

For the forthcoming complexity analysis, the structure of the
generating matrix of the vector code in $\LL^N$ obtained by
evaluations at the basis $\mathcal B$ will be used.

\begin{pro}\label{prop:struct_gen_mat_Kummer}
  Consider the code $\RM_{\LL /\K}(r,m)$ as the $\LL$-subspace of $\LL^N$
  obtained by evaluation of the $\TH$--polynomials at the elements of
  $\mathcal B$.  Denote by $\Gm(r,m) \in \LL^{k \times N}$ the
  generator matrix of $\RM_{\LL/\K}(r,m)$ whose rows correspond to the
  elements of $\Gm$ of $\TH$--degree $\leq r$ sorted w.r.t the reverse
  lexicographic order.
  Then,
  \[
    \Gm (r,m) =
    \begin{pmatrix}
      \Gm(r,m-1) & \alpha_m \Gm(r,m-1) \\
      \Gm(r-1,m-1) & -\alpha_m \Gm(r-1,m-1) \\      
    \end{pmatrix}.
  \]
\end{pro}

\begin{IEEEproof}
  Any element of the code is obtained by the evaluation at the elements of
  $\mathcal{B}$ of a $\TH$--polynomial
  \[F = \sum_{\gg \in \GG} f_\gg \gg = \sum_{\gg \in \GG_{m-1}} f_\gg \gg + \left(\sum_{\gg \in \GG_{m-1}} f_{\gg \theta_m}\gg\right)\theta_m.\]
  Hence
  \[
    F = F_0 + F_1 \theta_m,
  \]
  where $F_0, F_1 \in \LL [\GG_{m-1}]$ of respective $\TH$--degrees
  $\leq r$ and $\leq r-1$.  Since
  $\mathcal B = \mathcal B_{m-1} \cup \mathcal{B}_{m-1}\alpha_m$, denoting by
  $\text{ev}_{\mathcal{B}}$ the evaluation at the elements of $\mathcal B$,
  we get:
  \begin{align*}
    \text{ev}_{\mathcal B}(F) &= (\text{ev}_{\mathcal{B}_{m-1}}(F_0) +
                                \text{ev}_{\mathcal{B}_{m-1}}(F_1) ~|~ \alpha_{m}\text{ev}_{\mathcal{B}_{m-1}}(F_0)  - \alpha_{m} \text{ev}_{\mathcal{B}_{m-1}}(F_1))\\
                              &=  (\text{ev}_{\mathcal{B}_{m-1}}(F_0)  ~|~ \alpha_{m}\text{ev}_{\mathcal{B}_{m-1}}(F_0) )\ +\ ( \text{ev}_{\mathcal{B}_{m-1}}(F_1)~|~ - \alpha_{m} \text{ev}_{\mathcal{B}_{m-1}}(F_1)),
  \end{align*}
  which yields the result.
\end{IEEEproof}

Similar to $\Gm (r,m)$ in
Proposition~\ref{prop:struct_gen_mat_Kummer}, the code also admits a
parity check matrix $\Hm(r,m)$ with a recursive structure, which will
be used in the sequel for decoding in calculating the syndrome in the
fastest possible way. Therefore, we are interested in a similar
structure on parity--check matrices.

\begin{pro}\label{pro:recursive_dual}
  The code
  $\RM_{\LL/\K}(r,m)^{\perp}$ has a generator matrix
  $\Hm(r,m)$ with the recursive structure
\begin{equation}\label{eq:dual_kummer}
  \Hm(r,m) =
  \begin{pmatrix}
    \Hm (r,m-1) & \alpha_m^{-1}\Hm (r,m-1)\\
    \Hm (r-1,m-1) & -\alpha_m^{-1}\Hm (r-1,m-1)
  \end{pmatrix}. 
\end{equation}
\end{pro}

\begin{proof}
Reasoning by induction on $m$, we prove that
$\Gm (r,m) \Hm(r,m)^\top = 0$.
\end{proof}

\begin{rem}\label{rem:C(s,s)}
  Also, for the forthcoming fast syndrome calculation, we need to have
  a recursive structure for $\Hm (m,m-1)$ whose rows actually span
  $\LL^N$ but which will appear as a sub-block of $\Hm (m,m)$. For
  this, we take the same descriptions as (\ref{eq:dual_kummer}) and
  (\ref{eq:dual_AS}) respectively by setting
  $\Hm(m, m-1) \eqdef \Hm (m-1,m-1)$.  Therefore, the description of
  the matrices (\ref{eq:dual_kummer}) and (\ref{eq:dual_AS}) holds
  even for $r \geq m$. This remark also holds in the Artin--Schreier
  case to follow.
\end{rem}

\subsubsection{The Artin--Schreier case}

\begin{pro}[Recursive structure, the Artin--Schreier case]\label{AS-structure}
  Let $\text{Char } \K = 2$ and $\LL/\K$ be a Galois extension with
  $\Gal(\LL/K) \cong (\ZZ/{2\ZZ})^m$ and $r<m$ be nonnegative
  integers. Then $\RM_{\TH}(r,m)$ has the following recursive
  structure. If $r=-1$ then it is the zero code, composed only with the
  $2^m \times 2^m$ zero matrix. Otherwise, it is defined as:
\begin{equation}\label{eq:rec_struct_AS_case}
\left\{\scriptstyle{\begin{pmatrix}
{\Am_0+\Bm_0 } &\quad &{a_m( \Am_1+  \Bm_1) +\Bm_0}\\ 
 && \\
{ \Am_1+ \Bm_1} &\quad &{\Am_0 +\Am_1 + \Bm_0}
\end{pmatrix}} \quad \colon \quad \begin{aligned}
    & \Am_i \in \RM_{{\LL_{m-1}}/\K}(r, m-1), \notag \\&\Bm_i \in \RM_{{\LL_{m-1}}/\K}(r-1, m-1)
\end{aligned}\right\},
\end{equation}

where $\LL = \K (\alpha_1, \dots, \alpha_m)$ and $a_m=\alpha_m^2+\alpha_m$.
\end{pro}

\begin{IEEEproof}
  The proof is very similar to that of Proposition~\ref{rec_structure}
  with the slight difference that $\theta_m(\alpha_m) = \alpha_m+1$.
\end{IEEEproof}

\begin{pro}
  In the same context as Proposition~\ref{prop:struct_gen_mat_Kummer}
  but in the Artin--Schreier case, the ``canonical'' generator matrix
  $\Gm(r,m)$ of $\RM_{\LL / \K}(r,m) \subseteq \LL^N$ has the shape:
  \[
    \Gm (r,m) =
    \begin{pmatrix}
      \Gm (r, m-1)  & \alpha_m \Gm (r, m-1) \\
      \Gm (r-1, m-1)  & (\alpha_m+1) \Gm (r-1, m-1) \\      
    \end{pmatrix}.
  \]
  Moreover, in the same context as
  Proposition~\ref{pro:recursive_dual}, this code has a parity--check
  matrix with shape,
  \begin{equation}\label{eq:dual_AS}
  \Hm(r,m) =
\begin{pmatrix}
    (\alpha_m+1)\Hm (r,m-1) & \Hm (r,m-1)\\
    \alpha_m\Hm (r-1,m-1) & \Hm (r-1,m-1)
  \end{pmatrix}.
\end{equation}
\end{pro}

\begin{rem}\label{rem:parity-check_AS}
  The parity--check matrices can also be deduced from
  \cite[Cor.~52]{ACLN}. This statement claims that the vector
  representation of the code $\RM_{\LL /\K}(r, m)$ associated to a
  system of generators $\mathbf{\TH}$ and expressed in a basis
  $\mathcal B$ has a dual $\RM_{\LL/\K}(m-r-1,m)$ associated to
  generators $\mathbf{\TH}^{-1}$ and expressed in the dual basis
  $\mathcal{B}^*$  with respect to
the trace inner product $(a,b)\mapsto \text{Tr}_{\LL/\K}(ab)$.

  Since $\GG$ is of $2$--torsion, $\mathbf{\TH} = \mathbf{\TH}^{-1}$.
  Next, one can the following.
  \begin{itemize}
  \item In the Kummer case,
    $\mathcal{B}_1^* = (\frac{1}{2},\frac{1}{2\alpha_1})$ and
    $\mathcal{B}_m^* = \frac{1}{2}\mathcal{B}_{m-1}^* \cup
    \frac{1}{2}\mathcal{B}_{m-1}^* \alpha_m^{-1}$.  This permits to recover
    the expression~(\ref{eq:dual_kummer}).
  \item In the Artin--Schreier case, with
    $\mathcal{B}_2 = (1, \alpha_1, \alpha_2, \alpha_1 \alpha_2)$ we have
    $\mathcal{B}_2^* = ((\alpha_1+1)(\alpha_2+1), \alpha_2+1, \alpha_1+1, 1)$
    and $\mathcal{B}_m^* = \mathcal{B}_{m-1}^* (\alpha_m+1) \cup \mathcal{B}_{m-1}^*$. This
    permits to recover~(\ref{eq:dual_AS}).
  \end{itemize}
\end{rem}

 \section{Decoding}\label{decoding}

\subsection{The recursive decoding in the Hamming case}  
Let us quickly recall the behaviour of the decoder for binary
Reed--Muller codes $\RM_2 (r,m)$ of order $r$ and type $m$ in Hamming
metric, resting on the recursive structure mentioned in Lemma
\ref{recursivestructure}. See for instance \cite{D06}. The algorithm
is recursive: we assume to be able to decode any code
$\RM_2 (r',m')$ for either $r' < r$ or $m'<
m$.  Let
\[\yv = (\uv + \ev_l ~|~ \uv + \vv + \ev_r) \in \F_2^{2^m}\]
 be the received vector
such that $(\uv~|~\uv+\vv) \in \RM_2(r,m)$ and
$(\ev_l ~|~ \ev_r) \in \F_2^{2^m}$ has Hamming weight less or equal to
$2^{m-r-1}-1$. By \emph{folding} the received word \emph{i.e.}
by summing up its two halves, we get
\[(\uv + \ev_l) + (\uv+\vv+\ev_r) = (\vv + \ev_l +\ev_r) \quad \text{with}
  \quad
\wt(\ev_l + \ev_r) \le \wt(\ev_l ~|~ \ev_r) \le 2^{m-r-1}-1.\]

Since $\vv \in \RM_2(r-1,m-1)$, which has minimum distance $2^{m-r}$,
a recursive call to the algorithm permits to recover $\vv$ by decoding
$\vv+\ev_l+\ev_r$. Therefore, we now have to recover $\uv$ from
$(\uv +\ev_l ~|~ \uv + \ev_r)$, where
\[ \text{either}\quad
\wt(\ev_l) \le \frac{2^{m-r-1}-1}{2} \quad \text{or} \quad
\wt(\ev_r) \le \frac{2^{m-r-1}-1}{2}\cdot
\]
Therefore, at least one of the two recursive
calls respectively applied to $\uv+\ev_l$ and $\uv+\ev_r$ will succeed to recover $\uv$.

\subsection{A difficulty when switching to the rank case}\label{ss:difficulty_in_rank}
The aforementioned recursive decoding cannot be adapted to the rank
setting in a straightforward manner. Indeed, in the step of recovering
the part $\uv$, we crucially use the fact that in Hamming metric, when
cutting a weight $t$ vector in two halves $\ev = (\ev_l ~|~\ev_r )$,
then $w_{\text{H}}(\ev) = w_{\text{H}}(\ev_l) + w_{\text{H}}(\ev_r)$.
This property almost never holds in rank metric.
Indeed take $t < N/2$ and a rank $t$ matrix $\Em \in \K^{N\times N}$.
When splitting $\Em$ into two halves
\[
  \Em =
  \begin{pmatrix}
    \Em_l &|& \Em_r
  \end{pmatrix}
\]
with $\Em_l, \Em_r \in \K^{N \times \frac{N}{2}}$, then the most likely situation is that both
$\Em_l, \Em_r$ have rank $t$.

Therefore, in order to adapt the Hamming metric recursive decoder to the rank setting, we actually had to circumvent two difficulties:
\begin{enumerate}
\item find a matrix analogue of folding that permits to get rid of the contribution of the Reed--Muller code with the least minimum distance, namely  the matrices $\Am_0, \Am_1 \in \RM_{\LL/\K}(r, m-1)$ in~(\ref{eq:rec_struct_AS_case});
\item address the aforementioned issue, \emph{i.e.} once the folded codeword is decoded, \emph{i.e.}
  once $\Bm_0, \Bm_1 \in \RM_{\LL/\K}(r-1, m-1)$ are recovered we have to use a completely different approach to recover $\Am_0, \Am_1$.
\end{enumerate}

\subsection{Overview of our algorithm}
We give a recursive decoding algorithm for binary Reed--Muller
codes with rank metric for the Kummer case using the recursive structure identified in Proposition~\ref{rec_structure}.
Suppose we are given
     \begin{equation*}
 \Ym = \underbrace{\begin{pmatrix}
\Am_0+\Bm_0  &\quad &{a_m}( \Am_1- \Bm_1)\\ 
 && \\
\Am_1+ \Bm_1  &\quad &{\Am_0 -\Bm_0 }
\end{pmatrix}}_{\Cm \,\in\, \RM_{\LL/\K}(r,m)} + \underbrace{\begin{pmatrix}
\Em_{00}  &\quad &\Em_{01}\\ 
 && \\
\Em_{10}  &\quad &\Em_{11}
\end{pmatrix}}_{\Em \,\in \, {\K}^{2^m \times 2^m}},
\end{equation*}
where $\rk (\Em) \le 2^{m-r-1}-1$. Since the code $\RM_{\LL/\K}(r,m)$
has minimum distance $2^{m-r}$, the error's rank is less than half the
minimum distance.

\subsubsection*{Recursivity}
 To describe the recursive algorithm, we assume that we can decode up to half the minimum distance of any rank metric Reed--Muller code $\RM_{\LL/\K'}(r',m')$
over a field $\K'$ such that $\K \subseteq \K' \subseteq \LL$ with any
parameters $r' \leq r$ and $m'<m$.

\medskip

\subsubsection*{Overview of the algorithm} First, we briefly
describe the main steps of our decoding algorithm.

\begin{itemize}
\item \textit{Step 1. Folding $\Ym$.} We apply an operation that does
  not increase the error's rank and permits to get rid of the
  $\Am_i$'s and reduce to a decoding instance consisting in correcting
  $2^{m-r-1}-1$ errors for $\RM_{\LL / \K (\alpha_m)}(r-1,
  m-1)$. 
\item \textit{Step 2. Recursive calls.} Recursively calling the algorithm
  permits to recover $\Bm_0, \Bm_1$ and  get partial information on
  $\Em$.
\item \textit{Step 3. Recovery of the $\Am_i$'s.} Under some
  assumption on $\Em$, the partial information we got on the error
  permits to recover the $\Am_i$'s by solving linear systems.
\end{itemize}

\medskip

\subsection{Full description and proof of correctness of the algorithm}

\subsubsection{{Step 1. Folding}}

In the Hamming case, the decoding algorithm starts by folding the
received vector to get rid of the contribution of $\uv$, \emph{i.e.}
of the contribution of $\RM_2 (r, m-1)$ which is the one with the
smallest minimum distance. Similarly, here we look for an operation
that will not increase the rank of the error term while getting rid of
{ $\Am_0, \Am_1 \in \RM_{\LL/\K(\alpha_m)}(r,m-1)$}.

Denote by
$\II \in \K^{2^{m-1} \times 2^{m-1}}$ the identity matrix. The
following operation will be the analogue of the folding in the
Hamming case

\begin{align}
\label{reduction1}
  \begin{pmatrix}
    \frac{1}{{\a_m}} \II  & \II
  \end{pmatrix}
  \Ym 
  \begin{pmatrix}
    \II \\ - \frac{1}{{\a_m}} \II
  \end{pmatrix} &= 
  \frac{2}{{\a_m}}\Bm_0 + 2 \Bm_1  
   + \underbrace{\frac{1}{{\a_m}} \Em_{00} - \frac{1}{{a_m}} \Em_{01} + \Em_{10} - \frac{1}{{\a_m}} \Em_{11}}_{\eqdef\ \Em'};
\end{align}

\begin{rem}
It is worth noting that if the matrices $\Am_i, \Bm_i$ and $\Em_{ij}$ have their entries in $\K$, then the matrices obtained in reductions
  \eqref{reduction1} and \eqref{reduction2} have their entries in
  $\K (\alpha_m)$. Also, the operations on $\Em$ by left-and-right multiplication do not increase the rank, \emph{i.e.}, $\rk (\Em') \leq \rk (\Em)$.
\end{rem}

\begin{rem}
The folding operation could also have been  
\begin{align}
\label{reduction2}
  \begin{pmatrix}
-\frac{1}{{\a_m}} \II  & \II
\end{pmatrix} 
\Ym 
\begin{pmatrix}
\II \\  \frac{1}{{\a_m}} \II
\end{pmatrix} &= 
2 \Bm_1 - \frac{2}{{\a_m}}\Bm_0 
-  \frac{1}{{\a_m}} \Em_{00} - \frac{1}{{a_m}} \Em_{01} + \Em_{10} + \frac{1}{{\a_m}} \Em_{11},
\end{align}
which is nothing but the conjugate of (\ref{reduction1}). The use of
(\ref{reduction2}) is equivalent to that of (\ref{reduction1}).
\end{rem}

\subsubsection{{Step 2. A recursive call to recover $\Bm_0,\Bm_1$}}\label{sec:recursive}
Reduction \eqref{reduction1} yields the sum of an element of
$\RM_{\LL/\K(\alpha_m)}(r-1,m-1)$ and an error term of rank
$\rk (\Em')\leq \rk (\Em)\leq 2^{m-r-1}-1$, which is less than half
the minimum distance of $\RM_{\LL/\K(\alpha_m)}(r-1,m-1)$. Based on
our assumption that Reed--Muller codes of parameters $r'\leq r$ and
$m'<m$ can be decoded up to half their minimum distance, a recursive
call to the algorithm applied on the expression (\ref{reduction1})
permits to recover
\[
\frac{2}{{\a_m}}\Bm_0 +2 \Bm_1.
\]
Next, since $\Bm_0, \Bm_1$ have their entries in $\K$, they are
immediately deduced from above by extracting the $\K$ part and the
$\alpha_m$ part exactly as a complex matrix is decomposed into its
real and imaginary part.

\medskip

\noindent \textbf{Initialization of the recursive process.} In case
$r = -1$, then we are decoding the zero code and hence our decoder
will immediately return $0$.  In case $r=0$, then $\Bm_0, \Bm_1$ are
elements of $\RM_{\LL / \K (\alpha_m)}(-1, m-1)$ and hence are both
equal to $0$. Therefore, for $r = 0$, we immediately have access to
the $\Bm_i$'s. In this situation the folding should however be
computed: \eqref{reduction1} will yield the linear combination $\Em'$
of the $\Em_{ij}$'s, which is useful for Step 3 to follow.

\medskip

\subsubsection{{Step 3. Recovery of $\Am_0,\Am_1$}}\label{subsec:Step3}
Removing $\Bm_0, \Bm_1$ from $\Ym$, we
consider the new decoding problem:
 \begin{align}\label{eq:Ytilde}
\widetilde{\Ym} = \underbrace{\begin{pmatrix}
\Am_0  &\quad &{a_m} \Am_1\\ 
 && \\
 \Am_1  &\quad &{\Am_0 }
\end{pmatrix}}_{\widetilde{\Cm} \,\in\, \RM_{\LL/\K}(r,m)} + \underbrace{\begin{pmatrix}
\Em_{00}  &\quad &\Em_{01}\\ 
 && \\
\Em_{10}  &\quad &\Em_{11}
\end{pmatrix}}_{\Em \,\in \, {\K}^{2^m \times 2^m}}.
\end{align}

As explained in Section~\ref{ss:difficulty_in_rank}, to recover
$\Am_0,~\Am_1$, the technique used in the Hamming case does not
apply. However, the previous step permitted us also to recover from
(\ref{reduction1}) the following matrix:
\begin{align}
\small
 &\Em' \eqdef {\frac{1}{{\a_m}} \Em_{00} - \frac{1}{{a_m}} \Em_{01} + \Em_{10} - \frac{1}{{\a_m}} \Em_{11} }, \label{eq:E'}\end{align}
Here we make the following assumption.

\begin{assumption}\label{as:1}
The matrix
  $\Em' \in \K (\alpha_m)^{\frac N 2 \times \frac N 2}$ has
  rank $t$ and so do their foldings and so on until depth $r+1$,
  \emph{i.e.} we assume any iterated folding of $\Em$ of size
  $2^{m-i}\times 2^{m-i}$ with $i \leq r+1$ to have rank
  $t$.
\end{assumption}

\begin{rem}
  This assumption is actually what happens ``most of the time''.
  When applied
  to finite fields we analyse in Section~\ref{ss:proba} the
  probability that the folding of such a matrix $\Em$ keeps the same
  rank and we prove that this probability is exponentially close to
  $1$.
  The case of infinite fields is discussed later on in
  Section~\ref{ss:error_model} where we show that, under some choice of
  error model, we can bound from below the failure probability and make
  it exponentially close to $1$.
\end{rem}

Now, consider
\begin{equation}\label{eq:squeeze_Y}
\setlength{\abovedisplayshortskip}{3pt}
\setlength{\abovedisplayshortskip}{3pt}
\widetilde{\Ym}
\begin{pmatrix}
    \II \\ -\frac{1}{\alpha_m}\II
\end{pmatrix}
=
\begin{pmatrix}
    \Am_0 - \alpha_m \Am_1 \\
    \Am_1 - \alpha_m^{-1} \Am_0
\end{pmatrix}
+
\underbrace{\begin{pmatrix}
    \Em_{00} - \alpha_m^{-1} \Em_{01} \\
    \Em_{10} - \alpha_m^{-1} \Em_{11} 
\end{pmatrix}}_{\Fm}.
\end{equation}
The bottom part 
\begin{equation}\label{eq:def_F1}
   \Am_1 - \alpha_m^{-1} \Am_0 + \Fm_1, \quad \text{where} \quad \Fm_1 \eqdef  (\Em_{10} - \alpha_m^{-1} \Em_{11})
\end{equation}
gives an instance of the decoding problem for the code
$\RM_{\LL/\K (\alpha_m)}(r, m-1)$.  Moreover, since
\begin{equation}\label{eq:EFE'}
\small
    \Fm = \Em
  \begin{pmatrix}
    \II \\ - \frac{1}{\alpha_m}\II
  \end{pmatrix}
  \quad \text{and}\quad
  \Em' =
  \begin{pmatrix}
    \frac{1}{\alpha_m}\II & \II  
  \end{pmatrix} \Em
  \begin{pmatrix}
    \II \\ - \frac{1}{\alpha_m}\II
  \end{pmatrix}
  =\begin{pmatrix}
    \frac{1}{\alpha_m}\II & \II  
  \end{pmatrix}  \Fm,
\end{equation}
then \( t = \rk(\Em) \geq \rk (\Fm) \geq \rk (\Em') \) and, thanks to
Assumption~\ref{as:1}, all these matrices have rank $t$. Next, the
right--hand equation of (\ref{eq:EFE'}) entails that the row space of
$\Em'$ is contained in that of $\Fm$ and, these row spaces are
actually equal because they both have dimension $t$.  Thus, the row
space of the matrix $\Fm_1$ (defined in \eqref{eq:def_F1}) is contained in
that of $\Em'$.

Since $\Em'$ is known, we know its row space
$V \subseteq \K(\alpha)^{\frac N 2}$ and the previous discussion
asserts that $V$ contains the row space of $\Fm_1$. Similarly to the
Hamming case, with this partial knowledge of the row space (that can
be regarded as a rank metric counterpart of the support in the Hamming
metric) of $\Fm_1$, the decoding boils down to the resolution of a
linear system and permits to recover $\Am_1 - \alpha_m^{-1} \Am_0$.
This decoding approach while knowing the row space is well--known in
the literature (see, for instance \cite[\S~IV.1]{GRS16} or
\cite[\S~III.1]{AGHT18}) and can be regarded as the rank counterpart
of \emph{erasure decoding}. To ease the complexity analysis given in
Section~\ref{sec:complexity} we give an \emph{ad hoc} description of
the procedure in the subsequent paragraph. Note that $V$ has dimension
at most $t = 2^{m-r-1}-1$ which is less than the minimum distance of
$\RM_{\LL/\K(\alpha_m)}(r,m-1)$ which guarantees the uniqueness of the
solution of the decoding.

Finally, once $\Am_1 - \alpha_m^{-1} \Am_0$ is found, since both
$\Am_0, \Am_1$ have their entries in $\K$, they can be immediately
deduced.  

\subsubsection{Erasure decoding: further details on the recovery of $\Am_0, \Am_1$}\label{sec:erasure}
Let us
quickly explain how the decoding of (\ref{eq:squeeze_Y}) works. We
are given
\begin{equation}\label{eq:dec_prob_II}
  \widetilde{\Ym}_0 = (\Am_1 - \alpha_m^{-1} \Am_0) + \Fm_1 \ \in
  \K(\alpha_m)^{\frac N 2 \times \frac N 2},
\end{equation}
and we wish to recover
$(\Am_1 - \alpha_m^{-1} \Am_0) \in \RM_{\LL/\K (\alpha_m)}(r, m-1)$ with
$\rk (\Fm_1) \leq t$. In addition, we know a $\K (\alpha_m)$--subspace
$V \subseteq \K(\alpha_m)^{\frac N 2}$ of dimension $t$ that contains
the row space of $\Fm_1$. Then the decoding can be performed as
follows.  We compute a matrix
$\Rm \in \K(\alpha_m)^{t \times \frac N 2}$ (in row echelon form)
whose rows span $V$. Then, there exists 
$\Xm\in \K (\alpha_m)^{\frac N 2 \times t}$ such that
\begin{equation}\label{eq:f0}
  \Fm_1 = \Xm \Rm.
\end{equation}
Let us switch to a vector representation by representing elements of
$\K (\alpha_m)^{\frac N 2 \times \frac N 2}$ by vectors in
$\LL^{\frac N 2}$.  Note that the matrices in
$\K(\alpha_m)^{\frac N 2 \times \frac N 2}$ can be regarded either as
representations of elements of $\LL [\GG_{m-1}]$ in the basis
$\mathcal B_{m-1}$ or as elements of $\LL^{\frac N 2}$ obtained by
evaluating elements of $\LL [\GG_{m-1}]$ at the elements of
$\mathcal B_{m-1}$.  In vector representation, (\ref{eq:dec_prob_II})
becomes
\begin{equation}
  \label{eq:dec_prob_III}
  \widetilde{\yv}_0 = (\av_1 - \alpha_m^{-1} \av_0) + \fv_1 \ \in \LL^{\frac N 2},
  \qquad \text{where, from (\ref{eq:f0}),}\quad   \fv_1 = \xv \Rm,
  \qquad \text{for some } \xv \in \LL^{t}.  
\end{equation}
Since the space $V$ is known, the matrix
$\Rm$ is known and we aim to compute $\xv$.
For this sake, let $\Hm \in \LL^{(\frac N 2 - k) \times \frac N 2}$ be
a parity--check matrix of the code $\RM_{\LL/\K (\alpha_m)}(r, m-1)$
when regarded as a $k$--dimensional $\LL$-subspace of
$\LL^{\frac N 2}$. Then, $\xv$ can be computed as follows.

\begin{lem}\label{lem:uniqueness_erasures}
  The vector $\xv$ of \eqref{eq:dec_prob_III}
  is the unique solution of the linear system:
  \begin{equation}\label{eq:lin_sys}
    \Hm {\widetilde{\yv}_0}^{\top} = \Hm {(\xv \Rm)}^{\top} = \left(\Hm
      \Rm^{\top}\right) \xv^{\top}.
  \end{equation}
\end{lem}

\begin{proof}
  Only the uniqueness should be proven. Suppose that another
  $\xv' \neq \xv$ is solution as well.  Then
  $\Hm ((\xv-\xv')\Rm)^\top = 0$ or equivalently
  $(\xv-\xv')\Rm \in \RM_{\LL/\K(\alpha_m)}(r, m-1)$. Since $\Rm$ has
  full rank and $\xv \neq \xv'$ we deduce that
  $(\xv-\xv')\Rm \neq \mathbf 0$. Thus, as a nonzero element of
  $\RM_{\LL/\K(\alpha_m)}(r, m-1)$, we have
  \[\rk_{\K (\alpha_m)}(\xv-\xv')\Rm \geq 2^{m-1-r},\] while $\xv-\xv' \in \LL^{t}$ has length $t$
  and hence has $\K (\alpha_m)$--rank at most $t \leq 2^{m-1-r}-1$. Since
  $\Rm$ has entries in $\K(\alpha_m)$ then
  \[\rk_{\K(\alpha_m)} (\xv - \xv')\Rm \leq \rk_{\K (\alpha_m)}(\xv -
    \xv') \leq t \leq 2^{m-1-r}-1.\] A contradiction.
\end{proof}

\begin{algorithm}
\caption{\textsc{ErasureDecRM$({\Xm}, r,m, V)$}}
\begin{flushleft}
\textbf{Input:} A matrix $\Xm = \mathbf{C} + \mathbf{E} \in \K^{2^{m} \times 2^{m}}$, where $\mathbf{C} \in \textrm{RM}_{{\LL}/{\K}}(r,m)$ and $\rk(\mathbf{E}) = t \le 2^{m-r}-1$, and a vector space $V \subseteq {\K}^{2^{m}}$ which contains the row space of $\mathbf{E}$\\
\textbf{Output:} The codeword $\mathbf{C} \in \textrm{RM}_{{\LL}/{\K}}(r,m)$
\end{flushleft}
\hrule
\begin{flushleft}   
$\Hm \gets \text{ a parity-check matrix } \Hm(r,m) \text{ of } RM_{\LL/\K}(r,m)$ \hfill /* $\Hm \in \K^{(2^{m} - k) \times 2^{m}} \text{ where } k =\dim RM_{{\LL}/{\K}}(r,m)$  */ \\
$\mathbf{R} \gets$ a matrix in $\K^{t \times 2^m}$ whose rows span $V$\\ 
$\mathbf z \gets$ solution of $\Hm \mathbf{R}^{\top} \mathbf{z}^{\top}$\\
$\Em \gets$ matrix representation of $\ev \eqdef \mathbf z \Rm$ \\
\Return $\Xm - \Em$  \hfill /* {Returning $\Cm$} */
\end{flushleft}
\end{algorithm}

\begin{algorithm}
\caption{\textsc{DecodeRM}($\mathbf{Y}; r, m, \LL, \K$)}
\begin{flushleft}
\textbf{Input:}
\begin{itemize}
\item An $[\LL[\GG],k, d]_{\K}$-rank metric code $\CC (=\textrm{RM}_{\LL/\K}(r,m))$;
\item $\mathbf{Y} = \mathbf{C} + \mathbf{E} \in \K^{2^m \times 2^m},$ where $\mathbf{C} \in \CC,\, \rk(\mathbf{E}) \le 2^{m-r-1}-1$.
\end{itemize}
\textbf{Output:} Estimated codeword $\mathbf{C} \in \CC$ or \emph{decoding failure}
\end{flushleft}
\hrule
\begin{algorithmic}[1]
\If{$r=-1$}

\State{\Return $\mathbf{0} \in \K^{2^m \times 2^m}$}\\

\EndIf

\State $\Ym' \gets  \textbf{ Fold}(\Ym, a_m) \eqdef \begin{pmatrix}
    \frac{1}{\alpha_m} \II  & \II
  \end{pmatrix}
  \Ym 
  \begin{pmatrix}
    \II \\ - \frac{1}{\alpha_m} \II
  \end{pmatrix} 
$ \hfill /* $a_m = \alpha_m^2$ as given in definition of $\CC$*/\\

\State{$\frac{2}{{\sqrt{a}}} \Bm_0 + 2 \Bm_1 \gets \textbf{Decode}(\Ym',r-1,m-1)$}\\
\State{$\Bm_0, \Bm_1 \gets \frac{2}{{\sqrt{a}}} \Bm_0 + 2 \Bm_1$}\\
\State{$\widetilde{\Ym} \gets \Ym - \begin{pmatrix} \Bm_0 & - a\Bm_1\\ \Bm_1 & -\Bm_0 \end{pmatrix}$}\\
\State{$\widetilde{\Ym}_0 \gets \widetilde{\Ym} \begin{pmatrix} \II \\ - \frac{1}{\sqrt{a}}\II \end{pmatrix}$}\\
\State{$\Am_0, \Am_1$  or  ``\emph{Decoding failure}'' $\gets$ \textbf{ErasureDecRM}$(\widetilde{\Ym}_0, r',m')$}\\

\If{No \emph{Decoding failure}}

  \State{ \Return $\Cm \eqdef \begin{pmatrix} \Am_0 + \Bm_0 & a(\Am_1  - \Bm_1)\\
                       \Am_1 + \Bm_1 & \Am_0 - \Bm_0 \end{pmatrix}$}
\Else
   \State{\Return \emph{Decoding failure}}
\EndIf
\end{algorithmic}
\label{algo:decoder}
\end{algorithm}

\subsection{The Artin--Schreier case}
Suppose $\K$ has characteristic $2$, \emph{i.e.}, $\LL/\K$ is an Artin--Schreier extension. 
The received vector has the following form
     \begin{equation*}
       \Ym = \underbrace{\begin{pmatrix}
           \Am_0+\Bm_0  &\quad &{a_m}( \Am_1+ \Bm_1) +\Bm_0\\ 
            && \\
            \Am_1+ \Bm_1  &\quad &{\Am_0 +\Am_1+\Bm_0 }
\end{pmatrix}}_{\Cm \,\in\, \RM_{\LL/\K}(r,m)} + \underbrace{\begin{pmatrix}
\Em_{00}  &\quad &\Em_{01}\\ 
 && \\
\Em_{10}  &\quad &\Em_{11}
\end{pmatrix}}_{\Em \,\in \, {\K}^{2^m \times 2^m}},
\end{equation*}
where $\rk(\Em)\le 2^{m-r-1}-1$. Then, the approach is the same as in
the Kummer case with the exception that we get the following 
equation involving $\Bm_0$ and $\Bm_1$:
\[
\begin{pmatrix}
 \II  & \a_m \II
\end{pmatrix} \Ym \begin{pmatrix}
\II \\  \frac{1}{\a_m} \II
\end{pmatrix}
= \frac{1}{{\a_m}}\Bm_0 + \Bm_1 + \Em_{00} + \a_m
\Em_{10} + \frac{1}{{\a_m}} \Em_{01} + \Em_{11}.
\]
Next, a similar approach permits to
decode up to half the minimum
distance.

\subsection{Correctness}

\begin{thm}
  Let $-1\leq r\leq m$ be integers. Given an input \begin{equation}\label{eq:dec_pb}
    \Ym = \Cm + \Em
  \end{equation} where
  $\Cm \in \RM_{\LL/\K}(r,m)$ and $\rk (\Em) = t \leq 2^{m-r-1}-1$, then,
  under Assumption~\ref{as:1}, Algorithm~\ref{algo:decoder} returns
  $\Cm$.
\end{thm}

\begin{proof}
  Note first, that the since $\RM_{\LL/\K}(r,m)$ has minimum distance $2^{m-r}$, the
  constraint on $\rk (\Em)$ entails that $\Cm$ is the unique solution to the decoding problem
  \eqref{eq:dec_pb}.

  We prove the result by induction on $r$.
  For $r = -1$ then $\Cm = \mathbf 0$ and hence the algorithm succeeds.  
  For $r\geq 0$, then, folding $\Ym$ yields the instance
  (\ref{reduction1}) of decoding for $\RM_{\LL/\K}(r-1,m-1)$ where the
  error term $\Em'$ has rank below half the minimum distance of
  $\RM_{\LL/\K}(r-1,m-1)$. Thus, by induction, the decoder applied to
  the folding of $\Ym$ returns $\frac{2}{\a_m}\Bm_0 + \Bm_1$.  Then, we
  recover $\Bm_0, \Bm_1$ and can
  reduce the decoding problem to the ``erasure decoding''
  problem~(\ref{eq:dec_prob_II}), where, according to Assumption~\ref{as:1}
  the row space of the error term $\Fm_1$ is known.
  This erasure decoding problem is solved by computing the
  solution of~(\ref{eq:lin_sys}) which, from
  Lemma~\ref{lem:uniqueness_erasures}, is unique.
\end{proof}

 \section{Complexity analysis}\label{sec:complexity}
Here we give a complete analysis of the algorithm's complexity.
Since we are handling calculations in various field
extensions of $\K$ we will take as a reference the operations in
$\K$. Given two extensions $\K_1, \K_2$, we denote by
$\mul(\K_1, \K_2)$ the number of operations in $\K$ that costs a
multiplication of an element of $\K_1$ by an element of $\K_2$. Note
that \begin{equation}\label{eq:naive_cost_of_mult} [\K_1:\K] + [\K_2 :
  \K] \leq \mul (\K_1,\K_2) \leq [\K_1:\K][\K_2 : \K]
\end{equation}
since such multiplication involves more than
$ [\K_1:\K] + [\K_2 : \K] $ elements of $\K$ and that the naive
multiplication costs $[\K_1:\K][\K_2 : \K]$ operations in $\K$.  When
$\K_1 = \K_2$, we simply use $\mul (\K_1)$ to denote the number of operations in $\K$ that a multiplication of two elements in $\K_1$ costs. Following \cite[Corollary 11.11]{GG13}, we
will assume that if $\K_1$ is an extension with $[\K_1 : \K] = \ell$, then
\begin{equation}\label{eq:Mul(N)}
  \mul (\K_1) = \OO(\ell \log \ell).
\end{equation}
Finally, at some places, we will be interested specifically in the multiplication by $\alpha_m$.

\begin{lem}\label{lem:mult_by_alpham}
  Recall that $N = [\LL : \K]$.
  The multiplication of an element of $\LL$ by $\alpha_m$ costs
  $\OO(N)$ operations in $\K$.
\end{lem}

\begin{proof}
  Let us start with the Kummer case.  An element $v$ of $\LL$ is
  written $v = v_0 + v_1 \alpha_m$ for
  $v_0, v_1 \in \LL_{m-1} = \K (\alpha_1, \dots, \alpha_{m-1})$.
  Thus, $v_0, v_1$ are represented as elements of $\K^{\frac N 2}$ and
  \[
    \alpha_m v = a_m v_1 + \alpha_m v_0.
  \]
  In practice, it consists in a permutation of the entries of the
  element of $\K^N$ representing $v$ and then the multiplication of
  the $\frac N 2$ entries of $v_1$ by $a_m$ which is in $\K$.
  This yields $\OO (N)$ operations in $\K$.

  In the Artin--Schreier case, the multiplication by $\alpha_m$:
  \[
    \alpha_m (v_0 + v_1 \alpha_m) = v_1 a_m + (v_0 + v_1) \alpha_m
  \]
  which leads to the same asymptotic complexity estimate.
\end{proof}
Before
starting the analysis, we focus on a specific routine which is used in
Step 3: the calculation of \emph{syndromes}. That is to say products
$\Hm \yv^\top$ where
$\Hm \in \LL^{(N -\dim_\LL \RM_{\LL/\K} (r, m)) \times N}$ is the
structured parity--check matrix for $\RM_{\LL/\K} (r,m)$ as introduced
in Proposition~\ref{pro:recursive_dual}. Taking advantage of this
structure permits a faster computation compared to arbitrary matrix
products.

\subsection{Fast syndrome calculation}
As mentioned in Sections~\ref{sec:erasure}, Step 3, presented in
Section~\ref{subsec:Step3}, requires the computation of products
$\Hm {\widetilde{\yv}_0}^\top$ and $\Hm \Rm^\top$, where $\Hm$ is a
parity--check matrix of an RM code, $\widetilde{\yv}_0 \in \LL^{N/2}$
and $\Rm \in \K(\alpha_m)^{\frac N 2 \times t}$.

A naive version of the calculation of (for instance)
$\Hm \widetilde{\yv}_0^\top$ gives a cost in $\OO (N^2 \mul(\LL))$.
However, we can take advantage of the recursive structure of $\Hm$ to
perform a fast syndrome calculation.

  \begin{pro}\label{prop:complexity_of_syndrome_calculation}
    Let $\Hm$ be a matrix of the shape $\Hm (r, m)$ as defined in
    Equation~(\ref{eq:dual_kummer}) of Proposition~\ref{pro:recursive_dual}. Let $\K'$ be some
    intermediary field $\K \subset \K' \subset \LL$. Let
    $\yv \in (\K')^N$, then the product $\Hm \yv^\top$ can be computed
    in $\OO(N \mul (\K',\LL))$ operations in $\K$.
\end{pro}

\begin{proof}
  We prove the result in the Kummer case, the Artin--Schreier case
  being proved in the very same manner.  Let us first consider the
  complexity of computing $\Hm \yv^\top$ with $\yv \in \LL^N$.  Denote
  by $\compsyn (r,m)$ this complexity. From
  Proposition~\ref{pro:recursive_dual},
  \[
    \Hm \yv^\top = \Hm
    \begin{pmatrix}
      \yv_0^\top \\ \yv_1^\top
    \end{pmatrix} = 
  \begin{pmatrix}
    \Hm (r,m-1) & \alpha_m^{-1}\Hm (r,m-1)\\
    \Hm (r-1,m-1) & -\alpha_m^{-1}\Hm (r-1,m-1)
    \end{pmatrix}
        \begin{pmatrix}
      \yv_0^\top \\ \yv_1^\top
    \end{pmatrix}.
  \]
  The calculations of $\Hm(r,m-1)\yv_0^\top$ and
  $\Hm (r,m-1) \yv_1^\top$ cost $2 \compsyn(r, m-1)$. Then we need to
  multiply $\Hm (r,m-1) \yv_1^\top$ by
  $- \alpha_m^{-1} = -\frac{\alpha_m}{a_m}$ which, from
  Lemma~\ref{lem:mult_by_alpham} represents a cost of $\OO(N)$
  operations in $\K$ per entry, hence a cost in $\OO(N^2)$ operations in $\K$.
  Since the rows $\Hm (r-1, m-1)$ are rows of $\Hm(r,m-1)$, the
  vectors $\Hm(r-1,m-1)\yv_0^\top$ and
  $\alpha_m^{-1}\Hm (r-1,m-1) \yv_1^\top$ can be deduced for free
  from the previous calculations. There only remains to compute the sums
  \[
    \Hm(r,m-1)\yv_0^\top- \alpha_m^{-1}\Hm (r,m-1) \yv_1^\top \qquad \text{and}
    \qquad
        \Hm(r-1,m-1)\yv_0^\top+ \alpha_m^{-1}\Hm (r-1,m-1) \yv_1^\top,
      \]
      which also costs $\OO (N^2)$ operations in $\K$. 
      Thus, we have the recursive formula:
      \[
        \compsyn(r,m) \leq 2\compsyn (r, m-1) + \kappa N^2
      \]
      for some constant $\kappa$.
      Therefore,
      \[
        \compsyn (r,m) \leq \kappa \left(N^2 + 2
          {\left(\frac{N}{2}\right)}^2 + \cdots + 2^{m-r-1}
          {\left(\frac{N}{2^{m-r-1}} \right)}^2 \right) + 2^{m-r} \compsyn
        (r,r).
      \]
Next, from Remark~\ref{rem:C(s,s)},
      \begin{align*}
        \compsyn (r,m) \leq \kappa \Bigg(N^2 + \cdots &+ 2^{m-r-1}
          {\left(\frac{N}{2^{m-r-1}} \right)}^2 \Bigg) \\ &+ 2^{m-r}
        \kappa \left({\left(\frac{N}{2^{m-r}} \right)}^2 + \cdots +
        2^{r-1} {\left(\frac{N}{2^{m-1}} \right)}^{2} +2^{r} \compsyn (0,0) \right).
      \end{align*}
      Then, $\compsyn (0,0) = \mul (\K', \LL)$.  This yields an overall
      cost
      \[
        \compsyn (r,m) \leq \kappa' N^2 + N \mul(\K',\LL)
      \]
      for some other constant $\kappa'$. Then, from
      (\ref{eq:naive_cost_of_mult}), we deduce that
      \(
        \compsyn (r,m) = \OO (N \mul (\K',\LL)).
      \)
\end{proof}

\subsection{Complexity analysis of the decoding algorithm}
Denote by $\compdec (r,m,\K)$ the complexity of decoding
$\RM_{\LL/\K}(r,m)$ up to half the minimum distance, \emph{i.e.}
errors of rank up to $2^{m-r-1}-1$. We aim to evaluate this
quantity. 
\subsubsection{Complexity of Step  1}\label{subsec:cost_step1}
The first step of the decoding consists in folding $\Ym$. Thus,
consists in summing up $\frac N 2 \times \frac N 2$ sub-blocks of
$\Ym$.  Since $\Ym$ has entries in $\K$ multiplications by $\alpha_m$
are for free: we represent any element of $\K (\alpha_m)$ by
$u + \alpha_mv$ for $u,v \in \K$.

Thus, this step represents a cost in $\kappa (N/2)^2$ operations in $\K$ for
some positive constant $\kappa$.  

\subsubsection{Complexity of Step 2}\label{subsec:cost_step2}
The second step consists in a recursive call of the decoder on
decoding instances for $\RM_{\LL/\K (\alpha_{m})}(r-1,m-1)$, leading
to a cost of
\[
  \compdec (r-1, m-1, \K (\alpha_{m})).
\]

\subsubsection{Complexity of Step 3}
The third step consists in the recovery of $\Am_0, \Am_1$.  If it
boils down to linear algebra, the recursive structure of the codes
permits to perform structured linear algebra much faster than with
generic algorithms.
Indeed, according to Section~\ref{sec:erasure}, the decoding procedure
consists in
\begin{enumerate}
\item Computing a matrix $\Rm \in \K (\alpha_m)^{t \times \frac N 2}$
  whose rows span the space $V$ containing the error's row space;
\item Computing the syndrome $\Hm \widetilde{y}_0$ and the matrix $\Hm \Rm^\top$;
\item Solving the linear system \eqref{eq:dec_prob_III}.
\end{enumerate}

Thus, in terms of computations, Step 3 consists in the following steps.
\begin{enumerate}[(3.a)]
\item\label{item:3a} From the matrix $\Em'$ that was obtained from
  Step 2 (see~(\ref{eq:E'})), compute a basis of its row space $V$ to
  get the matrix $\Rm$. This can be performed by putting $\Em'$ in row
  echelon form. This can be done in less than
\[
    \kappa (N/ 2)^2t^{\omega-2} \mul (\K (\alpha_m)) ,\ \ \text{for some }\kappa >0
  \]
  (see for
  instance \cite[\S~3.3]{JPS13}).  
\item From Proposition~\ref{prop:complexity_of_syndrome_calculation}, the calculation of $\Hm
  \widetilde{\yv}_0$ can be performed in less than \[\kappa (N/2) \mul (\LL)\quad \text{operations,
  for some }\kappa >0.\]
\item Similarly to the previous case, using fast syndromes
  calculations, the computation of the matrix $\Hm \Rm^\top$ corresponds to
  computing $t$ products of $\Hm$ by a vector in
  $\K(\alpha_m)^N$. From
  Proposition~\ref{prop:complexity_of_syndrome_calculation}, this
  yields a cost bounded by
  \[
    \kappa t (N/2)\mul(\K(\alpha_m), \LL)) \quad \text{operations in
    }\K, \text{ for some }\kappa >0.
  \]
\item Solve the linear system~(\ref{eq:lin_sys}) to get $\xv$.  This
  is an $(N/2) \times t$ linear system with entries in $\LL$, yielding
  a cost in $\OO (N t^{\omega - 1} \mul (\LL))$ operations. This gives a complexity bounded by
  \[
    \kappa (Nt^{\omega-1} \mul(\LL))  \quad \text{operations in
    }\K, \text{ for some }\kappa >0.
\]
$\frac N{2^s}$.
\item\label{item:3e} Compute $\fv_0 =\xv \Rm$ where $\xv \in \LL^t$ and deduce
  $\av_1 - \alpha_m^{-1} \av_0$. The calculation of $\fv_0$
 represents
  \[
    \OO (t^2N \mul(\K (\alpha_m), \LL))\qquad \text{operations in}\ \K.
  \]
  This can be improved by expanding
  $\xv$ into an $\frac N 2 \times t$ matrix $\mathbf{X}$ with entries in
  $\K$. The cost of the product $\mathbf{X}\Rm$ is bounded by
  \[
    \kappa (N/2)^2 t^{\omega-2} \mul(\K, \K (\alpha_m)),\ \text{for some }\kappa > 0.
  \]
\end{enumerate}

We now have all the elements to give formulas for the complexity of
our algorithm.

\begin{thm}\label{thm:decoding_binary}
Under Assumption~\ref{as:1},
the decoding of $\RM_{\LL/\K} (r,m)$ can be performed with a deterministic
algorithm in
\[
    \OO (t^{\omega - 1} N \mul(\LL))\ \ \text{operations in }\K,
  \]
  which gives $\OO (t^{\omega-1}N^2 \log N)$ operations in $\K$.
\end{thm}

\begin{IEEEproof}
Summarizing the costs of the various steps we get a
  recursive formula:
  \begin{align*}
    \compdec (r,m, \K) \leq \compdec (r-1 &, m-1, \K (\alpha_m)) + \\
                                            & \kappa \Big((N/2)^2 + (N/2)^2 t^{\omega - 2} \mul (\K (\alpha_m))+ (N/2) \mul (\LL)\\
                                            &+ t(N/2)
                                              \mul (\K(\alpha_m),\LL)
                                              + t^{\omega - 1}(N/2) \mul (\LL) + t^{\omega - 2}(N/2)^2 \mul(\K, \K (\alpha_m)) \Big),
  \end{align*}
  for some positive constant $\kappa$.  Since $\omega$ is known to be
  $\geq 2$, some terms are negligible, which
  gives
  \[
    \compdec (r,m, \K) \leq \compdec (r-1 , m-1, \K (\alpha_m)) + 
                                             \kappa \Big( \frac N 2 t^{\omega - 1}  \mul (\LL)  \Big).
  \]
After recursive calls up to depth $r+1$
  we get:
  \[
    \compdec(r,m,\K) \leq \compdec (-1, m-r-1, \K(\alpha_{m-r}, \dots,
    \alpha_m)) + \kappa  \left(\frac N 2 + \frac{N}{4} + \cdots
        + \frac{N}{2^{r+1}} \right)t^{\omega - 1} \mul(\LL).
                     \]
The cost of a recursive call for a Reed--Muller code of degree $-1$, \emph{i.e.} of dimension $0$
is nothing but the cost of Step 1, and hence is bounded by $\kappa (N^2/2^r)$.
This yields
\[  \compdec(r,m,\K) \leq \kappa \left(\frac{N^2}{2^r} + N t^{\omega-1} \mul(\LL) \right).
\]
Since $\mul(\LL) = \Omega(N \log N)$, the $N^2/2^r$ term in the right--hand side is negligible
compared to the other term, which leads to:
\[
  \compdec(r,m,\K) = \OO (t^{\omega - 1} N \mul (\LL)).
\]
This provides the complexity. Using~(\ref{eq:Mul(N)}), we deduce the $\OO (t^{\omega-1} N^2 \log N)$.
\end{IEEEproof}

\begin{rem}
  Note that our complexity formula differs from that of
  \cite{CP25b}. Indeed, compared to our former short paper, we changed
  the initialization step of the recursive decoder which originally
  involved the Dickson decoder \cite{CP25}. The new version is fully independent from
  the Dickson decoder and benefits from a much cheaper initialization step.

  Also, note that \cite[Rem.~III.8]{CP25b} refers to the Dickson decoder
  while claiming a complexity in $\OO (k t^\omega \mul (\LL))$. This claim
  was based on a first version of \cite{CP25} which was not fully correct. This has been
  corrected in the sequel. The true complexity of the Dickson decoder \cite{CP25}
  is discussed in the subsequent Section \ref{subsec:comparison}.
\end{rem}

\subsection{Comparison with the Dickson--based approach}\label{subsec:comparison}
The recursive algorithm we presented in this paper fares well compared
to the Dickson matrix--based decoding algorithm in
\cite[Thm~5.15]{CP25}.  Indeed, according to this reference, the
Dickson--based decoder costs requires the Gaussian elimination of a
rank $t$ matrix in $\LL^{N \times N}$ and $\OO (kN)$ evaluations of
elements of $\GG$ at elements of $\LL$. Even when neglecting the
Galois action part, this leads to a complexity of
$\OO (t N^2 \mul(\LL))$ operations in $\K$.

In comparison, our algorithm has a complexity of
$\OO (t^{\omega - 1} N  \mul (\LL))$ operations in $\K$, which,
since $t \leq N$ is always asymptotically better because $\omega$ can be
taken to be $< 3$.  Our algorithm is even particularly better in the
linear rate regime: $k = \Omega (N)$ where $t$ is sublinear in $N$.

\begin{rem}
  Getting into \cite[\S~5.5]{CP25}, one could argue that the Gaussian
  elimination part can be performed in
  $\OO (t^{\omega - 2} N^2 \mul (\LL))$ operations in $\K$. This claim
  is disputable since the aforementioned Gaussian elimination involves
  formal variables that are specialised ``on the fly'' during the
  decoding step. Even if one could achieve this
  $\OO (t^{\omega-2} N^2 \mul (\LL))$ operations for the decoding,
  our new algorithm would still remain the best one in the regime $t = \circ (N)$ and would
  have the same complexity as the Dickson approach in the regime $t = \Theta(N)$.
\end{rem}

 \section{Plotkin construction for matrix codes}\label{sec:5}

Inspired by the previous sections where an analogue of Plotkin
$(u ~|~ u+v)$ structure was identified for rank metric binary
Reed--Muller codes, we abstract this structure to provide a rank
analogue of Plotkin construction.  A particular interest of this
approach is that contrary to rank metric Reed--Muller codes that can
be defined only over infinite fields, we can here propose a structure
that can be applied to any rank metric codes. In particular our
construction works over finite fields, which is the context we
consider from now on even if one could define the construction in a
much broader generality. Thus, in the sequel the ground field will be
a finite field $\Fq$ with $q$ odd.

\medskip

\noindent \textbf{Caution.} In the previous section, the Reed--Muller
codes corresponded to codes of length $N$ \emph{i.e.}, $N \times N$
matrix codes. Moreover, following the notation of \cite{ACLN,CP25}, we
used $m$ for the number of generators of the Galois group which yields
$N = 2^m$. In the sequel, we go back to more usual notation on rank
metric codes by considering spaces of $m \times n$ matrices. Hence,
contrary to the previous section, $m$ (and also $2m$) will stand for
the number of rows. Later, in Section~\ref{ss:error_model} we
establish a connection with rank Reed--Muller codes and for this
specific subsection, and only in this one, we will go back to the
notation used in the previous sections.

\subsection{Construction and first properties}

\begin{defn}\label{def:rank_plotkin}
  Let $a \in \Fq$.  Let $\CC$ be an $[m \times n, k_1, d_1]_q$ code
  and $\DD$ be an $[m \times n, k_2, d_2]_q$ code. We define the
  Plotkin construction of $\CC$ and $\DD$ to be the $\Fq$-linear
  matrix rank metric code:
\begin{equation}
\CC \plotkin_a \DD \eqdef 
\left\{\begin{pmatrix}
{\Am_0+\Bm_0 } &\quad &{a(\Am_1-  \Bm_1)}\\ 
 && \\
{ \Am_1+ \Bm_1} &\quad &{\Am_0 -\Bm_0}
\end{pmatrix} \quad\colon\quad  \begin{aligned}
     &\Am_i \in \CC, \notag \\
     &\Bm_i \in \DD
\end{aligned}\right\} \subseteq \Fq^{2m \times 2n}.
\end{equation}
\end{defn}

\begin{rem}\label{Plotkin:AS}
  The rationale behind choosing odd $q$ is to adapt the decoding
  algorithm for Kummer extensions to finite fields. For $q$ even, we
  propose to consider the Plotkin construction of $\CC$ and $\DD$ to
  be the following $\Fq$-linear matrix rank metric code imitating the
  recursive structure appeared for the Artin--Schreier extensions:
\begin{equation}
\CC \plotkinAS_a \DD \eqdef 
\left\{\begin{pmatrix}
{\Am_0+\Bm_0 } &\quad &{a(\Am_1+  \Bm_1) + \Bm_0}\\ 
 && \\
{ \Am_1+ \Bm_1} &\quad &{\Am_0 +\Am_1 + \Bm_0}
\end{pmatrix} \quad\colon\quad  \begin{aligned}
     &\Am_i \in \CC, \notag \\
     &\Bm_i \in \DD
\end{aligned}\right\}.
\end{equation}
\end{rem}

\begin{pro}\label{pro:dim_plotkin}
  In the context of Definition~\ref{def:rank_plotkin},
  \[
    \dim_{\Fq} \CC \plotkin_a \DD = 2 (\dim_{\Fq} \CC + \dim_{\Fq} \DD).
  \]
\end{pro}

\begin{IEEEproof}
  It suffices to observe that the maps 
\[
  (\Am_0, \Bm_0) \mapsto (\Am_0 + \Bm_0, \Am_0 - \Bm_0) \qquad
  \text{and} \qquad (\Am_1, \Bm_1) \mapsto (\Am_1+\Bm_1, a(\Am_1 -
  \Bm_1))
\]
are injective.  
\end{IEEEproof}

\begin{rem}
The dimension formula also holds true when $q$ is even and $ \CC \plotkinAS_a \DD$ defined as in Remark \ref{Plotkin:AS}. Indeed, as the map
\[
(\Am_0, \Bm_0, \Am_1, \Bm_1) \mapsto (\Am_0 + \Bm_0, \Am_1+\Bm_1, a(\Am_1 + \Bm_1)+ \Bm_0, \Am_0 + \Am_1 + \Bm_0)
\]
is injective.
\end{rem}

\begin{lem}
For $q$ odd, the dual of the Plotkin construction $\CC \plotkin_a \DD$ of matrix rank metric codes $\CC$ and $\DD$ is the Plotkin construction of their duals. More precisely,
\[
(\CC \plotkin_a \DD)^{\perp} = \CC^{\perp} \plotkin_{a^{-1}} \DD^{\perp}.
\]
\end{lem}
\begin{IEEEproof}
  Using Proposition~\ref{pro:dim_plotkin} we observe that the
  dimensions coincide. Therefore, it suffices to establish that
  $\Tr(\Cm \Dm^{\top}) =0$ for any $\Cm \in \CC \plotkin_a \DD$ and
  $\Dm \in \CC^{\perp} \plotkin_{a^{-1}} \DD^{\perp}$.
For $ \Am_i \in \CC, \Bm_i \in \DD$ and $ \Sm_i \in \CC^{\perp}, 
     \Um_i \in \DD^{\perp}$,
      \begin{align*}
  \Tr (\Cm \Dm^{\top}) &=  \Tr \left(\begin{pmatrix}
{\Am_0+\Bm_0 } &\quad &{a(\Am_1-  \Bm_1)}\\ 
 && \\
{ \Am_1+ \Bm_1} &\quad &{\Am_0 -\Bm_0}
\end{pmatrix}
     \begin{pmatrix}
{\Sm_0^{\top}+\Um_0^{\top} } &\quad &{\Sm_1^{\top} +  \Um_1^{\top}}\\ 
 && \\
{a^{-1} (\Sm_1^{\top} - \Um_1^{\top})} &\quad &{\Sm_0^{\top} -\Um_0^{\top}}
\end{pmatrix}\right)\\
&= \Tr\Bigg((\Am_0 + \Bm_0)(\Sm_0^{\top} + \Um_0^{\top}) + (\Am_1 - \Bm_1)( \Sm_1^{\top} - \Um_1^{\top}) + (\Am_1 + \Bm_1)(\Sm_1^{\top} + \Um_1^{\top}) + (\Am_0 - \Bm_0)(\Sm_0^{\top} - \Um_0^{\top})\Bigg)\\
&=2\Tr\Bigg( \Am_0 \Sm_0^{\top} +  \Bm_0 \Um_0^{\top} + \Am_1 \Sm_1^{\top} +  \Bm_1 \Um_1^{\top}   \Bigg) = 0.
\end{align*}
 It completes the proof of the lemma.
\end{IEEEproof}

\begin{rem}
  Note that the result does not hold when $q$ is even for the codes
  defined in Remark~\ref{Plotkin:AS}. Still, one can prove that:
  \[
    {(\CC \plotkinAS_a \DD)}^\perp = \left\{
      \begin{pmatrix}
        \Sm_0 + \Um_0 + \Um_1 & \quad & \Sm_1 + \Um_1
        \\ & \\
        a(\Sm_1 + \Um_1)+\Um_0+\Um_1 & \quad & \Sm_0+\Sm_1+\Um_0+\Um_1 
      \end{pmatrix} ~:~ \begin{aligned}
     &\Sm_i \in \CC^\perp, \notag \\ 
     &\Um_i \in \DD^\perp
\end{aligned}
      \right\}.
  \]
  To obtain this structure, we started from rank Reed--Muller in the
  Artin Schreier case and expressed the dual (which is also a rank
  Reed--Muller code) in the dual basis $\mathcal{B}_m^*$ as described in
  Remark~\ref{rem:parity-check_AS}.
\end{rem}

\subsection{Decoding} In the sequel, for convenience's sake we
restrict to the case of square matrices, \emph{i.e.}, $m=n$.

\subsubsection{When $a$ is not a square in $\Fq$}\label{sss:a_non_square}
If $a$ is not a square in $\Fq$, then the decoding of such a code can
be performed by running in a same fashion as in Section~\ref{decoding}
for binary rank metric Reed--Muller codes.

Precisely, suppose that we are given when given two codes
$\CC \subset \Fq^{m\times n}$ and $\DD \subset \Fq^{m\times n}$
equipped with decoders such that $\DD \otimes \F_{q^2}$ can correct
errors of rank $t$ and $\CC \otimes \F_{q^2}$ can correct ``erasures''
of rank $t$, \emph{i.e.} can decode when the error term has rank
$\leq t$ and its row space is contained in a known space $V$ of
dimension $t$. Suppose that we are given
\[
  \Ym = \Cm + \Em,\quad \text{where}\quad \Cm \in \CC \plotkin_a \DD, \quad
  \text{and}\quad \rk \Em \leq t.
\]
Denote
\[
  \Cm =
  \begin{pmatrix}
    \Am_0 + \Bm_0 & a(\Am_1 - \Bm_1) \\
    \Am_1 + \Bm_1 & \Am_0 - \Bm_0
  \end{pmatrix}
  \quad \text{and} \quad
  \Em =
  \begin{pmatrix}
    \Em_{00} & \Em_{01} \\ \Em_{10} & \Em_{11}
  \end{pmatrix}.
\]
Choose a square root $\sqrt a$ of $a$ in $\F_{q^2}$.
Then one can perform the decoding as follows:
\begin{enumerate}
\item Fold $\Ym$ as
  \begin{equation}\label{eq:fold_plotkin}
\mathbf{Fold}(\Ym) \eqdef  \begin{pmatrix}
    \frac{1}{{\sqrt{a}}} \II_m  & \II_m
  \end{pmatrix}
  \Ym 
  \begin{pmatrix}
    \II_n \\ - \frac{1}{{\sqrt a}} \II_n
  \end{pmatrix} = 
  \frac{2}{{\sqrt{a}}}\Bm_0 + 2 \Bm_1  
   + \underbrace{\frac{1}{{\sqrt{a}}} \Em_{00} - \frac{1}{a} \Em_{01} +  \Em_{10} - \frac{1}{a} \Em_{11}}_{\eqdef\ \Em'}.
    \end{equation}
\item The error term $\Em'$ has rank $\leq t$. Thus, apply the decoder of
  $\DD$ to the above folding to recover
  $\frac{2}{\sqrt a} \Bm_0 + 2\Bm_1$ and hence the pair $\Bm_0, \Bm_1$
  (immediate since the $\Bm_i$'s have entries in $\Fq$ while
  $\sqrt a \in \F_{q^2} \setminus \Fq$).
\item With the knowledge of $\Bm_0, \Bm_1$ we are reduced to the following
  decoding problem:
  \begin{equation}\label{eq:Ytilde_2}
    \widetilde{\Ym} =
    \begin{pmatrix}
      \Am_0 & a \Am_1 \\ \Am_1 & \Am_0
    \end{pmatrix} + \Em.
  \end{equation}
  Next, using a partial folding of $\widetilde \Ym$, we get the following decoding problem
  \begin{equation}\label{eq:fold_Ytilde}
    \widetilde \Ym     \begin{pmatrix}
      \mathbf{I}_n \\ -\frac{1}{\sqrt a}\mathbf{I}_n
                       \end{pmatrix} =
                       \begin{pmatrix}
    \Am_0 - \sqrt a \Am_1 \\
    \Am_1 - \frac{1}{\sqrt{a}} \Am_0
                       \end{pmatrix} +
    \begin{pmatrix}
    \Em_{00} - \frac{1}{\sqrt{a}} \Em_{01} \\
    \Em_{10} - \frac{1}{\sqrt{a}} \Em_{11} 
\end{pmatrix}                   
\end{equation}
and under Assumption~\ref{as:2} below (which is very similar to Assumption~\ref{as:1}
but without also involving recursive calls)
the row space of the error term is known. Applying our erasure
decoder to the lower part:
\[
  \Am_1 - \frac{1}{\sqrt{a}} \Am_0 + (\Em_{10} - \frac{1}{\sqrt{a}} \Em_{11}),
\]
we can recover $\Am_1 - \frac{1}{\sqrt a} \Am_0$ and immediately deduce
the pair $\Am_0, \Am_1$ which yields the whole $\Cm$.
\end{enumerate}

\begin{assumption}\label{as:2}
  In the process above, the folded error $\Em'$ of (\ref{eq:fold_plotkin})
  has the same rank as $\Em$.
\end{assumption}

\subsubsection{When $a$ is a square in $\Fq$}
If $a$ is a square in $\Fq$, then one needs to adopt a slightly
different approach that requires to involve twice more calls to the decoders but
without working in field extensions. Namely, let us also fix a square root $\sqrt{a} \in \Fq$
of $a$. Given $\Ym = \Cm + \Em$, the decoding will run as follows:
\begin{enumerate}
\item Fold $\Ym$ in two different manners:
  \begin{align*}
    \mathbf{Fold}_1(\Ym) &\eqdef  \begin{pmatrix}
    \frac{1}{{\sqrt{a}}} \II_m  & \II_m
  \end{pmatrix}
  \Ym 
  \begin{pmatrix}
    \II_n \\ - \frac{1}{{\sqrt a}} \II_n
  \end{pmatrix} = 
  \frac{2}{{\sqrt{a}}}\Bm_0 + 2 \Bm_1  
    + \frac{1}{{\sqrt{a}}} \Em_{00} - \frac{1}{a} \Em_{01} + \Em_{10} - \frac{1}{\sqrt{a}} \Em_{11}; \\
    \mathbf{Fold}_2(\Ym) &\eqdef  \begin{pmatrix}
    -\frac{1}{{\sqrt{a}}} \II_m  & \II_m
  \end{pmatrix}
  \Ym 
  \begin{pmatrix}
    \II_n \\  \frac{1}{{\sqrt a}} \II_n
  \end{pmatrix} = 
  -\frac{2}{{\sqrt{a}}}\Bm_0 + 2 \Bm_1  
    - \frac{1}{{\sqrt{a}}} \Em_{00} - \frac{1}{a} \Em_{01} +\Em_{10} + \frac{1}{\sqrt{a}} \Em_{11}.
  \end{align*}
\item Apply the decoder of $\DD$ to the two above foldings to recover
  $\frac{2}{\sqrt{a}}\Bm_0 + 2\Bm_1$ and $-\frac{2}{\sqrt{a}}\Bm_0 + 2\Bm_1$ from which
  we deduce the pair $\Bm_0, \Bm_1$.
\item We deduce another decoding problem which is the same as (\ref{eq:Ytilde_2}).
  Next, we involve partial foldings as in~(\ref{eq:fold_Ytilde}), however, here, since we
  are no longer working over a quadratic extension, we need two partial foldings:
  \[
    \widetilde{\Ym}
    \begin{pmatrix}
      \II_n \\ -\frac{1}{\sqrt{a}}\II_n
    \end{pmatrix} \quad \text{and}  \quad
    \widetilde{\Ym}
    \begin{pmatrix}
      \II_n \\ \frac{1}{\sqrt{a}}\II_n
    \end{pmatrix}
  \]
  As in Section~\ref{sss:a_non_square}, taking the bottom parts of these partial foldings,
  leads to two decoding problems whose error term has a known row space. Namely:
  \[
    \Am_1 - \frac{1}{\sqrt{a}}\Am_0 + (\Em_{10} - \frac{1}{\sqrt{a}}\Em_{11}) \quad
    \text{and}\quad
    \Am_1 + \frac{1}{\sqrt{a}}\Am_0 + (\Em_{10} + \frac{1}{\sqrt{a}}\Em_{11}).
  \]
  As in Section~\ref{sss:a_non_square}, applying the erasure decoder for $\CC$ to the two
  instances above,  yields $\Am_1 \pm \frac{1}{\sqrt{a}}\Am_0$ from which we deduce the
  pair $\Am_0, \Am_1$.
\end{enumerate}
The process above is summarised in Algorithm~\ref{algo:Plotkin}.

\begin{algorithm}
  \caption{\textsc{\textsc{DecodePlotkinConstruction}($\mathbf{Y}; \CC, \DD, a$)}
  \hfill /* Case $a$ is a square in $\Fq$ */}\label{algo:Plotkin}
\begin{flushleft}
  \textbf{Input:} An $[m \times m, k_1]_q$ code $\CC$ and an $[m \times m, k_2]_q$ code $\DD$
  with decoder correcting $t$ errors for $\DD$ and $t$ erasures for $\CC$\\
  A matrix $\mathbf{Y} \in \mathbb{F}_{q}^{2m \times 2m} = \Cm + \Em$ for some $\Em \in \Fq^{2m \times 2m}$ of rank $\leq t$\\
\textbf{Output:} Either $\mathbf{C}$ or \emph{decoding failure}
\end{flushleft}
\begin{algorithmic}[1]
\State $\sqrt a \gets $ a square root of $a$ in $\Fq$.
\State \textbf{Foldings:}  \[\Ym' \gets \begin{pmatrix}
    \frac{1}{{\sqrt{a}}} \II_m  & \II_m
  \end{pmatrix}
  \Ym 
  \begin{pmatrix}
    \II_n \\ - \frac{1}{{\sqrt{a}}} \II_n
  \end{pmatrix}  \quad \text{and} \quad \Ym'' \gets \begin{pmatrix}
    - \frac{1}{{\sqrt{a}}} \II_m  & \II_m
  \end{pmatrix}
  \Ym 
  \begin{pmatrix}
    \II_n \\  \frac{1}{{\sqrt{a}}} \II_n
  \end{pmatrix}\]
\State \textbf{Recovering the part of $\DD$:}  
\begin{align*}
&\frac{2}{{\sqrt{a}}} \Bm_0 + 2 \Bm_1 \gets ErrorDec_{\DD}(\Ym') \text{ and } - \frac{2}{{\sqrt{a}}} \Bm_0 + 2 \Bm_1 \gets ErrorDec_{\DD}(\Ym'')\\
&\Bm_0, \Bm_1 \gets \pm \frac{2}{{\sqrt{a}}} \Bm_0 + 2 \Bm_1
\end{align*}
 \State \textbf{Removing $\Bm_0, \Bm_1$ from $\Ym$:} $\widetilde{\Ym} \gets \Ym - \begin{pmatrix} \Bm_0 & - a\Bm_1\\ \Bm_1 & -\Bm_0 \end{pmatrix}$
 \State \textbf{Recovering the part of $\CC$:} $
 \begin{pmatrix}
   \widetilde{\Ym}_{00} \\ \widetilde{\Ym}_{01}
 \end{pmatrix}
 \gets \widetilde{\Ym} \begin{pmatrix} \II_n \\ - \frac{1}{\sqrt{a}}\II_n \end{pmatrix}\qquad $
 and $\qquad
 \begin{pmatrix}
   \widetilde{\Ym}_{10} \\ \widetilde{\Ym}_{11}
 \end{pmatrix}
 \gets \widetilde{\Ym} \begin{pmatrix} \II_n \\  \frac{1}{\sqrt{a}}\II_n \end{pmatrix} $
 \State
 \[
   \Am_0 -\frac{1}{\sqrt{a}}\Am_1 \gets ErasureDec_{\CC}(\widetilde{\Ym}_{01}), \qquad
   \Am_0 + \frac{1}{\sqrt{a}}\Am_1 \gets ErasureDec_{\CC}(\widetilde{\Ym}_{11}), \qquad
 \]
\If{No \emph{Decoding failure}}
\State $\Am_0, \Am_1 \gets \Am_0 \pm \frac{1}{\sqrt{a}}\Am_1$   
   \State \Return $\begin{pmatrix} \Am_0 + \Bm_0 & a(\Am_1  - \Bm_1)\\
                       \Am_1 + \Bm_1 & \Am_0 - \Bm_0 \end{pmatrix}$
\Else
   \State{\Return \emph{Decoding failure}}
\EndIf
\end{algorithmic}
\end{algorithm}

\subsubsection{The Artin--Schreier case}
In characteristic $2$, the Artin--Schreier case can be decoded in a similar fashion by separating the cases when $a$ can be written as $\alpha^2+\alpha$ for $\alpha \in \F_{q^2}\setminus \Fq$ and
the case where $\alpha \in \Fq$.

\subsection{Probability analysis}\label{ss:proba}
According to Assumption~\ref{as:1}, for the algorithm to succeed, we
need the following \emph{folding} of $\Em$ to have same rank as
$\Em$. The folding has, up to scalar multiplication the following
form:
\[
 \text{Fold}(\Em) \eqdef\begin{pmatrix} \II_m & b \II_m
  \end{pmatrix}
  \Em   \begin{pmatrix}
    b'\II_n \\ \II_n
  \end{pmatrix}
\]
for some nonzero $b,b'$ which are either elements of $\Fq$ or elements $\Fqq$
whose squares are in $\Fq$, depending on whether $a$ is a square or a non square
in $\Fq$.

To estimate the probability
that $\text{Fold}(\Em)$ has rank $t$ we proceed in two steps.
First we evaluate the probability:
\[
  \Prob{\rk \left(\begin{pmatrix} \II_m & b \II_m
  \end{pmatrix}
  \Em\right) = t}.
\]
Then, we evaluate the overall probability. It turns out that the probability behaves
differently if $a$ is a square or a non square. Thus, we treat the two situations separately.

\subsubsection{If $a$ is a square in $\Fq$}\label{sss:proba_a_square}
Note first that the right kernel of $\begin{pmatrix} \II_m & b \II_m
\end{pmatrix}$ has dimension $m$. Denote by $K$ this kernel. Second,
note that since $\Em$ is a uniformly random $2m \times 2n$ matrix of
rank $t$, its column space $U$ is a uniformly random subspace of
$\Fq^{2m}$ of dimension $t$. For the rank not to collapse during the
first part of the folding we should have $U \cap K = \{0\}$. Here
we use the following statement.

  \begin{lem}\label{lem:esp}
    Let $K$ be a subspace of $\Fq^{2m}$ of dimension $m$ and $U$ be a
    uniformly random subspace of $\Fq^{2m}$ of dimension $t < m$.
    Then \[
      \Prob{U \cap K \neq 0} = \OO (q^{t-m-1}).\]
  \end{lem}

  \begin{IEEEproof}    
    Let $\Mm_U$ be an arbitrary $t \times 2m$ generator matrix of $U$,
    \emph{i.e.} a matrix whose rows span $U$. Then $\Mm_U$ is a
    uniformly random full--rank $t \times 2m$ matrix and for any
    $\mv \in \Fq^t \setminus \{0\}$ then $\mv \Mm_U$ is a uniformly
    random vector in $\Fq^{2m}\setminus \{0\}$. Thus,
    \[\Prob{\mv \Mm_U \in K} = \frac{q^m-1}{q^{2m}-1} = q^{-m}(1+ \circ (1)).\]  Therefore,
    \begin{align*}
      \Esp{|U \cap K \setminus \{0\}|} &= \sum_{\mv \in \Fq^t \setminus \{0\}}
                                         \Prob{\mv \Mm_U \in K}\\
                                       & =  (q^t-1)q^{-m}(1 + \circ (1))\\
                                       & =  q^{t-m}(1+ \circ (1))                         
    \end{align*}
    Next, if $U \cap K \neq \{0\}$ then it has cardinality $\geq q$
    and, by Markov inequality,
    \[
      \Prob{|U \cap K \setminus \{0\}| \geq q-1} \leq \frac{\Esp{|U \cap K \setminus \{0\}|}}{q-1}.
    \]
    Hence the result.
  \end{IEEEproof}
 
  As a consequence of this lemma, the probability that the first folding
  keeps rank $t$ is in $1- \OO (q^{t-m-1})$.
  For the second folding, assuming that it has rank $t$, then
  the obtained matrix has a row space $V$ of dimension $t$ which is
  nothing but the row space of $\Em$. Hence it is a uniformly random
  space of dimension $t$ in $\Fq^{2n}$.
  
  Denote by $K'$ the left kernel of $
  \begin{pmatrix}
    \II_n \\ b' \II_n
  \end{pmatrix}
$, then,
  according to Lemma~\ref{lem:esp},
  \[
    \Prob{V \cap K' \neq 0} = \OO (q^{t-n-1}).
  \]
  In summary, the probability that both folding have rank $t$ is
  $1 - \OO (q^{t-\min(m,n)-1})$.

\subsubsection{If $a$ is not a square in $\Fq$}
In this situation, then $b,b'$ are in $\Fqq \setminus \Fq$.  Moreover,
similarly to the case of rank metric Reed--Muller codes studied in
Section~\ref{Binary_rankRM} we assume $b^2 = {b'}^2 = a \in \Fq$. For
the first folding, one can do the same reasoning with the only
difficulty that the column space $U$ of $ \Em $ is in $\Fq^{2m}$ and
the right kernel of $\begin{pmatrix} \II_m & b \II_m
\end{pmatrix} $ is in $\Fqq^{2m}$. However, this situation is even better
as explains the following lemma.
  
  \begin{lem}
    Let $K$ be a subspace of $\Fqq^{2m}$ of dimension $m$.
    Let $U$ be
    a uniformly random subspace of dimension $t$ in $\Fq^{2m}$ and
    $U_2 \eqdef U \otimes \Fqq$. Then
    \[
      \Prob{U_2 \cap K \neq 0} = \OO (q^{2t-2m-2}).
    \]
  \end{lem}

  \begin{IEEEproof}
    Let $\Mm_{U}$ be an arbitrary $t \times 2m$ generator matrix of $U$.  Let
    $b \in \Fqq \setminus \Fq$.  Elements of $U_2$ are
    $\mv_1\Mm_U + b \mv_2 \Mm_U$ for $\mv_1, \mv_2 \in \Fq^{t}$.
    Similarly to the proof of Lemma~\ref{lem:esp}, $\mv_1\Mm_U + b \mv_2 \Mm_U$
    turns out to be a uniformly random element of $\F_{q^2}^{2m}$ and we deduce that
    \[
      \Prob{\mv_1\Mm_U + b \mv_2 \Mm_U \in K} = q^{-2m}(1+\circ (1)).
    \]
    Next,
    \begin{align*}
      \Esp{|U_2 \cap K \setminus \{0\}|} & = 
      \sum_{(\mv_1, \mv_2) \in \Fq^{2t} \setminus \{0\}} \Prob{\mv_1\Mm_U + b \mv_2 \Mm_U \in K} \\
      & = q^{2t-2m}.
    \end{align*}
    Finally, as for Lemma~\ref{lem:esp}, Markov inequality permits to
    conclude.
  \end{IEEEproof}

  For the second folding, the situation gets more complicated since
  the obtained matrix after the first folding is not a uniformly
  random matrix of rank $t$ in $\Fqq^{m \times 2n}$. Indeed its row
  space $V_2$ is actually $V \otimes \Fqq$ where $V$ is the row space
  of $\Em$. In short, its row space has a basis in $\Fq^{2n}$ which
  does not hold for any such rank $t$ matrix in $\Fqq^{m \times
    2n}$. Still, one can prove that its row space is uniform among the
  $t$--dimensional spaces of the form $V \otimes \Fqq$.  Then, the
  following statement will permit to conclude.

  \begin{lem}
    Let $V$ be a uniformly random subspace of dimension $t$
    of $\Fq^{2m}$. Let $K$ be the left kernel of $
    \begin{pmatrix}
      \II_n \\ b' \II_n
    \end{pmatrix}
    $. Then
    \[
      \Prob{ (V \otimes \Fqq) \cap K \neq \{0\}} = \OO (q^{2t -2n - 2}).
    \]
  \end{lem}

  \begin{proof}
    First let us prove that any nonzero element $\vv$ of the left kernel of $
    \begin{pmatrix}
      \II_n \\ b\II_n
    \end{pmatrix}
    $ satisfies $\vv = \vv_0 + b \vv_1$, were $\vv_0, \vv_1 \in \Fq^{2n}$ are linearly independent.
    Indeed suppose $\vv_0 = \lambda \vv_1$ for some $\lambda \in \Fq$. Then,
    \[
      \vv \begin{pmatrix}
      \II_n \\ b\II_n
          \end{pmatrix} = 0 \quad \Longrightarrow \quad
          (\lambda + b)\vv_1 \begin{pmatrix}
      \II_n \\ b\II_n
                             \end{pmatrix} = 0 \quad \Longrightarrow \quad
                             \vv_1 \begin{pmatrix}
      \II_n \\ b\II_n
                             \end{pmatrix} = 0.
    \]
  Recall that $\vv_1$ has entries in $\Fq$ and denote
  $\vv_{10}, \vv_{11}$ its left and right halves. That is to say
  $\vv_1 = (\vv_{10} ~|~ \vv_{11})$.
  Then, the above equality yields $\vv_{10}+b\vv_{11} = 0$ and since $\vv_{10}, \vv_{11}$
  have entries in $\Fq$ while $b \in \F_{q^2}\setminus \Fq$, we deduce that $\vv_1=0$
  and hence $\vv=0$.

  Now, let $\Mm \in \Fq^{t \times 2n}$ be a full rank matrix whose rows span $V$. An element of
  $V \otimes \F_{q^2}$ has the shape
  \[
    (\mv_0 + b \mv_1) \Mm \quad \text{where} \quad \mv_0, \mv_1 \in \Fq^{t}.
  \]
  Given a pair of linearly independent vectors
  $\mv_0, \mv_1 \in \Fq^{t}$, then $(\mv_0 + b \mv_1) \Mm$ is of the
  form $\vv_0 + b \vv_1$ with $\vv_0, \vv_1 \in \Fq^{2n}$ linearly
  independent. Moreover, $(\mv_0 + b \mv_1) \Mm$ is uniformly random among such vectors.
  Thus
  \begin{align*}
    \mathbb{P}\left( (\mv_0 + b \mv_1) \Mm
      \begin{pmatrix}
        \II_n \\ b \II_n
      \end{pmatrix} = 0
    \right) & = \frac{\left|\ker_{\text{Left}}
        \begin{pmatrix}
          \II_n \\ b \II_n
        \end{pmatrix}
              \right|}{\left| \{\vv = \vv_0 +b \vv_1 \in \F_{q^2}^{2n} ~|~ \vv_0,\vv_1 \in \Fq^{2n}\ \text{linearly independent}\} \right|} \\
    & = \frac{q^{2n}-1}{(q^{2n}-1)(q^{2n}-q)} = \frac{1}{q^{2n}-q}\cdot
  \end{align*}
  Therefore,
  \[
    \mathbb{E}\left(|V \otimes \F_{q^2} \cap K| \right) = \sum_{\mv_0, \mv_1} \mathbb{P}\left( (\mv_0 + b \mv_1) \Mm
      \begin{pmatrix}
        \II_n \\ b \II_n
      \end{pmatrix} = 0
    \right)
  \]
  where the sum is taken over all the pairs $\mv_0, \mv_1 \in \Fq^{t}$. Note that the above probability is $0$ if $\mv_1, \mv_2$ are collinear and $1/(q^{2n}-q)$ otherwise, which yields
  \[\mathbb{E}\left(|V \otimes \F_{q^2} \cap K| \right) =
    \frac{(q^t-1)(q^t-q)}{q^{2n}-q} = \OO(q^{2t-2n}).\]
  Finally, by Markov inequality,
  \[\mathbb{P}\left(V \otimes \F_{q^2} \cap K \neq 0 \right)
    = \mathbb{P}\left(|V \otimes \F_{q^2} \cap K \neq 0| \geq q^2 \right)
    = \OO (q^{2t-2n-2}.)
  \]
  \end{proof}

  With these lemmas, we deduce that
  $\Prob {\rk (\text{Fold}(\Em) < t)} = \OO (q^{2t - 2\min(m,n) - 2})$.
 
\subsubsection{Experimental results}
We tested the validity of Assumption~\ref{as:2} using \emph{SageMath}
\cite{sagemath}. For $a$ square,
\begin{itemize}
\item With parameters $q=23, m=16, t=4$ we
  generated a random rank $t$ matrix and folded it.  We ran 10,000 tries
  and none of them let to a folded matrix of rank $<4$.
\item With parameters $q=23, m=16, t=4$ we
  generated a random rank $t$ matrix and folded it.  We ran 1,000,000 tries
  and only one matrix got a folding of rank $<4$. 
\end{itemize}

\subsection{Complexity} The following statement summarises the previous discussions.

\begin{thm}\label{thm:plotkin}
  Let $\CC$ be an $[m \times n, k_1]_q$ code and $\DD$ an
  $[m \times n, k_2]_q$ code where $q$ is power of an odd prime. Suppose that
  for $\DD$, we can correct $t \leq \frac{\min (m,n)}{2}$ errors for some
  positive integer $t$ and that for $\CC$ we can correct $t$ erasures
  (\emph{i.e.}  errors of rank $t$ whose row or column space is
  known). Then, we have an algorithm taking as input
  \[
    \Ym = \Cm + \Em
  \]
  where $\Cm \in \CC \plotkin_a\DD$ and $\Em$ is uniformly random
  among the $2m \times 2n$ matrices with rank $t$, which returns the
  pair $(\Cm, \Em)$. Moreover,
  \begin{itemize}
  \item if $a$ is a square in $\Fq$, then the algorithm succeeds with
    probability $1 - \OO(q^{t-\min (m,n)-1})$ and has complexity
    $\OO (mn) + 2(C(t)+D(t))$ operations
    in $\Fq$, where $C(t)$ and $D(t)$ denote respectively the
    complexity of correcting $t$ erasures for $\CC$ and $t$ errors for
    $\DD$;
  \item if a is not a square in $\Fq$, then the algorithm succeeds with
    probability $1 - \OO(q^{2t-2\min (m,n)-2})$ and has complexity
    $\OO (mn) + C(t)+D(t)$ operations
    in $\Fq$, where $C(t)$ and $D(t)$ denote respectively the
    complexity of correcting $t$ erasures for $\CC\otimes \F_{q^2}$ and $t$ errors for
    $\DD \otimes \F_{q^2}$.
  \end{itemize}
\end{thm}

\begin{rem}
  Of course, exactly as in the case of rank metric binary Reed--Muller codes
  or in the spirit of \cite{MCT16,MCT17} in Hamming metric, this
  Plotkin--like construction can be iterated and decoded with a
  recursive decoder ultimately calling other decoders. 
\end{rem}

\subsection{An instantiation with Gabidulin codes} \label{exa:nonMRD}
The Plotkin construction gives a scheme to obtain efficiently decodable rank metric codes by combining a pair of rank metric codes with efficient decoders. As mentioned before, there are essentially three families of structured codes previously known to have efficient decoding algorithms and they are mostly $\Fqm$-linear: Gabidulin codes, LRPC codes and Simple codes. However, these codes as components of the Plotkin construction can provide new families of rank metric codes with efficient decoding algorithms. We see below such an example of Plotkin construction of a code by considering a pair $(\CC, \DD)$ of Gabidulin codes.

Let $\CC, \DD \subset \Fq^{m \times m}$ be two matrix representation
of Gabidulin codes over $\Fqm$ of respective $\Fqm$--dimensions
$k_1, k_2$.  Recall that these codes have respective minimum distances
$m-k_1+1$ and $m-k_2+1$. Moreover, Loidreau's algorithm \cite{Loid}
permits to decode $\DD$ up to $\frac{m-k_2}{2}$ in $\OO(m^2)$
operations in $\Fqm$ and it is possible to correct up to $m-k_1$
rank erasures for $\CC$.  Suppose that $m = 2k_1- k_2$ so that
\begin{equation}\label{eq:k1k2}
t \eqdef (m-k_1) = \frac{m-k_2}{2}\cdot
\end{equation}
Then, according to Proposition~\ref{pro:dim_plotkin} and
Theorem~\ref{thm:plotkin}, the code $\CC \plotkin_a \DD$ has $\Fq$--dimension
$2m(k_1+k_2)$ and benefits from a decoder correcting up to $t$ errors.

\begin{rem}
To conclude, note that $\CC \plotkin_a \DD$ contains some
  \[
    \begin{pmatrix}
      \Am_0 & (0) \\ (0) & \Am_0
    \end{pmatrix},
  \]
  with $\Am_0 \in \CC$ of weight $m-k_1+1$. Hence the minimum distance
  of $\CC \plotkin_a \DD$ is less than $2m - 2k_1 + 2$. If
  $m-k_2 > 2$, then (\ref{eq:k1k2}) entails that
  the latter minimum distance is strictly below the rank Singleton
  bound $2m - k_1 - k_2 + 1$. Therefore, $\CC \plotkin_a \DD$ is never
  MRD and hence cannot be equivalent to a Gabidulin code. The obtained
  family of rank--metric codes is hence a new one equipped with an
  efficient decoder.
\end{rem}

\subsection{About the error model}\label{ss:error_model} Back to the
case of Rank Reed--Muller case and suppose that $\K = \mathbb{Q}$ and
$\LL$ is some $(\mathbb{Z}/2\mathbb{Z})^m$--extension of $\mathbb{Q}$.
Take note that, for this subsection, we go back to the notation of the
previous sections where $m$ denotes the number of generators of the
Galois group. Our codes are spaces of $N \times N$ matrices where
$N = 2^m$. Let us consider the code $\RM_{\LL / \mathbb{Q}}(r,m)$ for some
nonnegative integer $m$ and let $t = 2^{m-r-1}-1$ be the number of errors we
expect to correct.

Let $p$ be a prime integer that splits totally in $\LL$;
Chebotarev density theorem asserts the existence of infinitely many
such primes. Here let us consider an error model where errors are
matrices $\Em$ of rank $\leq t$ whose entries lie in some interval
$[0, T]$ for some $T\geq  p$ and whose reduction $\overline{\Em}$
modulo $p$ also has rank $t$.

Under this error model, one can check that Assumption~\ref{as:1} holds
for a single folding if it holds modulo $p$. Thus, from
Section~\ref{sss:proba_a_square}, the modulo $p$ case gives a lower
bound for the probability that the folding of the error keeps the same
rank: it is at least $1 - \OO(p^{t-N-1})$. Since to decode
$\RM_{\LL/\mathbb{Q}}(r, m)$ we will perform $\OO(r)$ recursive calls,
Assumption~\ref{as:1} holds with a probability at least
\[1 - \OO(r p^{t-N-1}) = 1 - \OO(rp^{2^{m-r-1}-2^m-2}).\]
In short, the failure probability is doubly exponentially close to $0$.

 \section{Conclusion}
We proposed a novel algorithm for decoding
``binary--like'' rank metric Reed--Muller codes. This decoder rests on
the use of a recursive structure that we identified on these codes. In
its design, this decoder is extremely simple and basically boils down
to linear algebra.  If compared to the algorithm proposed in
\cite{CP25}, this new algorithm might fail on some rare instances, but has a
much better complexity.

Finally, our research on rank Reed--Muller codes led us to identify
what a natural rank analogue of the Plotkin $(u ~|~ u+v)$ construction
could be. This novel construction gives a generic frame to design new
families of efficiently decodable rank metric codes from old ones.  We
expect this novel operation on rank metric codes to open interesting
perspectives in the near future.

 \bibliographystyle{acm}

\begin{thebibliography}{10}

\bibitem{ASAM}
{\sc Abbe, E., Sberlo, O., Shpilka, A., and Ye, M.}
\newblock {R}eed--{M}uller codes.
\newblock {\em Foundations and Trends{\textregistered} in Communications and
  Information Theory 20}, 1--2 (2023), 1--156.
\newblock \url{http://dx.doi.org/10.1561/0100000123}.

\bibitem{AABBBBDDGLPRVZ22}
{\sc {Aguilar Melchor}, C., Aragon, N., Bettaieb, S., Bidoux, L., Blazy, O.,
  Bos, J., Deneuville, J.-C., Dion, A., Gaborit, P., Lacan, J., Persichetti,
  E., Robert, J.-M., V{\'{e}}ron, P., Z{\'{e}}mor, G., and Bos, J.}
\newblock {HQC}.
\newblock Round 4 Submission to the NIST Post-Quantum Cryptography Call, Oct.
  2022.
\newblock \url{https://pqc-hqc.org/}.

\bibitem{A19b}
{\sc {Aguilar Melchor}, C., Aragon, N., Bettaieb, S., Bidoux, L., Blazy, O.,
  Bros, M., Couvreur, A., Deneuville, J.-C., Gaborit, P., Z{\'e}mor, G., and
  Hauteville, A.}
\newblock Rank quasi cyclic {(RQC)}.
\newblock Second Round submission to NIST Post-Quantum Cryptography call, Apr.
  2020.
\newblock \url{https://pqc-rqc.org}.

\bibitem{ABCCGLMMMNPPPSSSTW20}
{\sc Albrecht, M., Bernstein, D.~J., Chou, T., Cid, C., Gilcher, J., Lange, T.,
  Maram, V., von Maurich, I., Mizoczki, R., Niederhagen, R., Persichetti, E.,
  Paterson, K., Peters, C., Schwabe, P., Sendrier, N., Szefer, J., Tjhai,
  C.~J., Tomlinson, M., and Wen, W.}
\newblock Classic {M}c{E}liece (merger of {Classic McEliece} and {NTS-KEM}).
\newblock \url{https://classic.mceliece.org}, Nov. 2022.
\newblock Fourth round finalist of the NIST post-quantum cryptography call.

\bibitem{ROLLO}
{\sc Aragon, N., Blazy, O., Deneuville, J.-C., Gaborit, P., Hauteville, A.,
  Ruatta, O., Tillich, J.-P., Z{\'e}mor, G., {Aguilar Melchor}, C., Bettaieb,
  S., Bidoux, L., Bardet, M., and Otmani, A.}
\newblock {ROLLO} (merger of {Rank-Ourobouros, LAKE and LOCKER}).
\newblock Second round submission to the NIST post-quantum cryptography call,
  Mar. 2019.

\bibitem{AGHRZ}
{\sc Aragon, N., Gaborit, P., Hauteville, A., Ruatta, O., and Z{\'e}mor, G.}
\newblock Low rank parity check codes: New decoding algorithms and applications
  to cryptography.
\newblock {\em IEEE Trans. Inform. Theory 65}, 12 (2019), 7697--7717.

\bibitem{AGHT18}
{\sc Aragon, N., Gaborit, P., Hauteville, A., and Tillich, J.-P.}
\newblock A new algorithm for solving the rank syndrome decoding problem.
\newblock In {\em 2018 {IEEE} International Symposium on Information Theory,
  {ISIT} 2018, Vail, CO, USA, June 17-22, 2018\/} (2018), IEEE, pp.~2421--2425.
\newblock \url{https://doi.org/10.1109/ISIT.2018.8437464}.

\bibitem{ACLN}
{\sc Augot, D., Couvreur, A., Lavauzelle, J., and Neri, A.}
\newblock Rank-metric codes over arbitrary {G}alois extensions and rank
  analogues of {R}eed--{M}uller codes.
\newblock {\em SIAM J. Appl. Algebra Geom. 5}, 2 (2021), 165--199.

\bibitem{ALR18}
{\sc Augot, D., Loidreau, P., and Robert, G.}
\newblock Generalized {G}abidulin codes over fields of any characteristic.
\newblock {\em Des. Codes Cryptogr. 86}, 8 (2018), 1807--1848.

\bibitem{BCCCDGKLNSST23}
{\sc Banegas, G., Carrier, K., Chailloux, A., Couvreur, A., Debris-Alazard, T.,
  Gaborit, P., Karpman, P., Loyer, J., Niederhagen, R., Sendrier, N., Smith,
  B., and Tillich, J.-P.}
\newblock {Wave}.
\newblock Round 1 Additional Signatures to the NIST Post-Quantum Cryptography:
  Digital Signature Schemes Call, June 2023.

\bibitem{BHLPRW}
{\sc Bartz, H., Holzbaur, L., Liu, H., Puchinger, S., Renner, J., and
  Wachter-Zeh, A.}
\newblock Rank-metric codes and their applications.
\newblock {\em Found. Trends Commun. Inf. Theory 19}, 3 (May 2022), 390--546.

\bibitem{BC24}
{\sc Berardini, E., and Caruso, X.}
\newblock {R}eed--{M}uller codes in the sum-rank metric.
\newblock {\em arXiv preprint arXiv:2405.09944\/} (2024).

\bibitem{CP25}
{\sc Couvreur, A., and Pratihar, R.}
\newblock Decoding rank metric {R}eed--{M}uller codes.
\newblock arXiv preprint
  \href{https://arxiv.org/abs/2501.04766}{arXiv:2501.04766}, 2025.

\bibitem{CP25b}
{\sc Couvreur, A., and Pratihar, R.}
\newblock {Recursive Decoding of Binary Rank Reed-Muller Codes and Plotkin
  Construction for Matrix Codes}.
\newblock In {\em Proc. IEEE Int. Symposium Inf. Theory - ISIT~2025\/} (2025),
  pp.~1--6.

\bibitem{DST19a}
{\sc {Debris-Alazard}, T., Sendrier, N., and Tillich, J.-P.}
\newblock Wave: A new family of trapdoor one-way preimage sampleable functions
  based on codes.
\newblock In {\em Advances in Cryptology - ASIACRYPT~2019, Part {I}\/} (Kobe,
  Japan, Dec. 2019), S.~D. Galbraith and S.~Moriai, Eds., vol.~11921 of {\em
  Lecture Notes in Comput. Sci.}, Springer, pp.~21--51.
\newblock \url{https://doi.org/10.1007/978-3-030-34578-5\_2}, doi =
  {10.1007/978-3-030-34578-5\_2}.

\bibitem{Del}
{\sc Delsarte, P.}
\newblock Bilinear forms over a finite field, with applications to coding
  theory.
\newblock {\em J. Combin. Theory Ser. A 25}, 3 (1978), 226--241.

\bibitem{D06}
{\sc Dumer, I.}
\newblock Soft-decision decoding of {R}eed--{M}uller codes: a simplified
  algorithm.
\newblock {\em IEEE Trans. Inform. Theory 52}, 3 (2006), 954--963.

\bibitem{Gab}
{\sc Gabidulin, E.~M.}
\newblock Theory of codes with maximum rank distance.
\newblock {\em Problemy Peredachi Informatsii 21}, 1 (1985), 3--16.

\bibitem{GPT91}
{\sc Gabidulin, E.~M., Paramonov, A.~V., and Tretjakov, O.}
\newblock Ideals over a non-commutative ring and their application in
  cryptology.
\newblock In {\em Workshop on the Theory and Application of of Cryptographic
  Techniques\/} (1991), Springer, pp.~482--489.

\bibitem{GHPT17}
{\sc Gaborit, P., Hauteville, A., Phan, D.~H., and Tillich, J.-P.}
\newblock Identity-based encryption from codes with rank metric.
\newblock In {\em Advances in Cryptology -- CRYPTO 2017\/} (Cham, 2017),
  J.~Katz and H.~Shacham, Eds., vol.~10403, Springer International Publishing,
  pp.~194--224.

\bibitem{GMRZ13}
{\sc Gaborit, P., Murat, G., Ruatta, O., and Z{\'e}mor, G.}
\newblock Low rank parity check codes and their application to cryptography.
\newblock In {\em Proceedings of the Workshop on Coding and Cryptography
  WCC'2013\/} (Bergen, Norway, 2013).

\bibitem{GRS16}
{\sc Gaborit, P., Ruatta, O., and Schrek, J.}
\newblock On the complexity of the rank syndrome decoding problem.
\newblock {\em IEEE Trans. Inform. Theory 62}, 2 (2016), 1006--1019.

\bibitem{GU19}
{\sc Geiselmann, W., and Ulmer, F.}
\newblock Skew {R}eed--{M}uller codes.
\newblock {\em Rings, Modules and Codes 727\/} (2019), 107--116.

\bibitem{GQ09}
{\sc Gow, R., and Quinlan, R.}
\newblock Galois theory and linear algebra.
\newblock {\em Linear Algebra and its Applications 430}, 7 (2009), 1778--1789.

\bibitem{JPS13}
{\sc Jeannerod, C.-P., Pernet, C., and Storjohann, A.}
\newblock Rank-profile revealing {G}aussian elimination and the {CUP} matrix
  decomposition.
\newblock {\em J. Symbolic Comput. 56\/} (2013), 46--68.

\bibitem{KLP68}
{\sc Kasami, T., Lin, S., and Peterson, W.}
\newblock {New generalizations of the Reed-Muller codes--I: Primitive codes}.
\newblock {\em IEEE Trans. Inform. Theory 14}, 2 (1968), 189--199.

\bibitem{KLP68b}
{\sc Kasami, T., Lin, S., and Peterson, W.}
\newblock {New generalizations of the Reed-Muller codes--II: Primitive codes}.
\newblock {\em IEEE Trans. Inform. Theory 14}, 2 (1968), 199--205.

\bibitem{KK08}
{\sc Koetter, R., and Kschischang, F.~R.}
\newblock Coding for errors and erasures in random network coding.
\newblock {\em IEEE Trans. Inform. Theory 54}, 8 (2008), 3579--3591.

\bibitem{L88}
{\sc Lachaud, G.}
\newblock Projective {R}eed--{M}uller codes.
\newblock In {\em Coding theory and applications ({C}achan, 1986)}, vol.~311 of
  {\em Lecture Notes in Comput. Sci.} Springer, Berlin, 1988, pp.~125--129.

\bibitem{L90}
{\sc Lachaud, G.}
\newblock The parameters of projective {R}eed--{M}uller codes.
\newblock {\em Discrete Math. 81}, 2 (1990), 217--221.

\bibitem{Loid}
{\sc Loidreau, P.}
\newblock A {W}elch-{B}erlekamp like algorithm for decoding {G}abidulin codes.
\newblock In {\em Coding and cryptography}, vol.~3969 of {\em Lecture Notes in
  Comput. Sci.} Springer, Berlin, 2006, pp.~36--45.

\bibitem{L17}
{\sc Loidreau, P.}
\newblock A new rank metric codes based encryption scheme.
\newblock In {\em Post-Quantum Cryptography~2017\/} (2017), vol.~10346 of {\em
  Lecture Notes in Comput. Sci.}, Springer, pp.~3--17.

\bibitem{LK05}
{\sc Lu, H.-F., and Kumar, P.}
\newblock A unified construction of space-time codes with optimal
  rate-diversity tradeoff.
\newblock {\em IEEE Trans. Inform. Theory 51}, 5 (2005), 1709--1730.

\bibitem{LTZ17}
{\sc Lunardon, G., Trombetti, R., and Zhou, Y.}
\newblock On kernels and nuclei of rank metric codes.
\newblock {\em J. Algebraic Combin. 46}, 2 (2017), 313--340.

\bibitem{MS77}
{\sc MacWilliams, F.~J., and Sloane, N. J.~A.}
\newblock {\em The theory of error-correcting codes}, vol.~16.
\newblock Elsevier, 1977.

\bibitem{MCT16}
{\sc M{\'a}rquez-Corbella, I., and Tillich, J.-P.}
\newblock Using {R}eed--{S}olomon codes in the $(u | u + v)$ construction and
  an application to cryptography.
\newblock In {\em 2016 IEEE International Symposium on Information Theory
  (ISIT)\/} (2016), IEEE, pp.~930--934.

\bibitem{MCT17}
{\sc M{\'a}rquez-Corbella, I., and Tillich, J.-P.}
\newblock Attaining capacity with iterated $(u | u + v)$ codes based on {AG}
  codes and {K}oetter--{V}ardy soft decoding.
\newblock In {\em 2017 IEEE International Symposium on Information Theory
  (ISIT)\/} (2017), IEEE, pp.~6--10.

\bibitem{MP}
{\sc Mart{\'{\i}}nez{-}Pe{\~{n}}as, U., and Pellikaan, R.}
\newblock Rank error-correcting pairs.
\newblock {\em Des. Codes Cryptogr. 84}, 1-2 (2017), 261--281.

\bibitem{M78}
{\sc McEliece, R.~J.}
\newblock {\em A Public-Key System Based on Algebraic Coding Theory}.
\newblock Jet Propulsion Lab, 1978, pp.~114--116.
\newblock DSN Progress Report 44.

\bibitem{N86}
{\sc Niederreiter, H.}
\newblock Knapsack-type cryptosystems and algebraic coding theory.
\newblock {\em Problems of Control and Information Theory 15}, 2 (1986),
  159--166.

\bibitem{Oreqpoly}
{\sc Ore, {\O}.}
\newblock On a special class of polynomials.
\newblock {\em Trans. Amer. Math. Soc. 35}, 3 (1933), 559--584.

\bibitem{Orenoncom}
{\sc Ore, {\O}.}
\newblock Theory of non-commutative polynomials.
\newblock {\em Ann. of Math.\/} (1933), 480--508.

\bibitem{RJB}
{\sc Renner, J., Jerkovits, T., and Bartz, H.}
\newblock Efficient decoding of interleaved low-rank parity-check codes.
\newblock In {\em 2019 XVI International Symposium "Problems of Redundancy in
  Information and Control Systems" (REDUNDANCY)\/} (2019), pp.~121--126.

\bibitem{R15}
{\sc Robert, G.}
\newblock {A new constellation for space-time coding}.
\newblock In {\em {WCC 2015}\/} (Paris, France, Apr. 2015), P.~Charpin,
  N.~Sendrier, and J.-P. Tillich, Eds., WCC 2015, {Anne Canteaut, Ga{\"e}tan
  Leurent, Maria Naya-Plasencia}.

\bibitem{R15a}
{\sc Robert, G.}
\newblock {\em {Codes de Gabidulin en caract{\'e}ristique nulle : application
  au codage espace-temps}}.
\newblock Theses, {Universit{\'e} de Rennes}, Dec. 2015.

\bibitem{Rot91}
{\sc Roth, R.~M.}
\newblock Maximum-rank array codes and their application to crisscross error
  correction.
\newblock {\em IEEE Trans. Inform. Theory 37}, 2 (1991), 328--336.

\bibitem{SB10}
{\sc Sidorenko, V., and Bossert, M.}
\newblock Decoding interleaved {G}abidulin codes and multisequence linearized
  shift-register synthesis.
\newblock In {\em 2010 IEEE International Symposium on Information Theory\/}
  (2010), IEEE, pp.~1148--1152.

\bibitem{SKK11}
{\sc Silva, D., and Kschischang, F.~R.}
\newblock Universal secure network coding via rank-metric codes.
\newblock {\em IEEE Trans. Inform. Theory 57}, 2 (2011), 1124--1135.

\bibitem{S91}
{\sc S{\o}rensen, A.~B.}
\newblock {Projective Reed--Muller codes}.
\newblock {\em IEEE Trans. Inform. Theory 37}, 6 (1991), 1567--1576.

\bibitem{sagemath}
{\sc {The Sage Developers}}.
\newblock {\em {S}ageMath, the {S}age {M}athematics {S}oftware {S}ystem
  ({V}ersion 9.5)}, 2022.
\newblock {\tt https://www.sagemath.org}.

\bibitem{GG13}
{\sc von~zur Gathen, J., and Gerhard, J.}
\newblock {\em Modern Computer Algebra}, 3rd~ed.
\newblock Cambridge University Press, 2013.

\bibitem{WL13}
{\sc Wu, B., and Liu, Z.}
\newblock Linearized polynomials over finite fields revisited.
\newblock {\em Finite Fields Appl. 22\/} (2013), 79--100.

\end{thebibliography}

\end{document}